%% file: head.tex
\newif\ifextended%
\extendedfalse%

\PassOptionsToPackage{x11names}{xcolor}
\ifextended
\documentclass[sigconf,authorversion,screen]{acmart}
\else
\documentclass[sigconf,screen]{acmart}
\fi

\usepackage[utf8]{inputenc}

\usepackage{booktabs,multirow}
\usepackage{comment}
\usepackage{url}
\usepackage[caption=false,font=footnotesize]{subfig}

\usepackage[colorinlistoftodos,prependcaption]{todonotes}
\usepackage{xspace}

\usepackage{bm}

\usepackage[nameinlink]{cleveref}

\usepackage{nameref}

\usepackage{enumitem}
\setlist{leftmargin=*,nosep}

\input{macros}

\author{Naoto Ohsaka}
\affiliation{\institution{NEC Corporation}\country{}}
\email{ohsaka@nec.com}
\orcid{0000-0001-9584-4764}

\author{Tatsuya Matsuoka}
\affiliation{\institution{NEC Corporation}\country{}}
\email{ta.matsuoka@nec.com}

\copyrightyear{2022}
\acmYear{2022}
\setcopyright{acmlicensed}\acmConference[WSDM '22]{Proceedings of the Fifteenth ACM International Conference on Web Search and Data Mining}{February 21--25, 2022}{Tempe, AZ, USA}
\acmBooktitle{Proceedings of the Fifteenth ACM International Conference on Web Search and Data Mining (WSDM '22), February 21--25, 2022, Tempe, AZ, USA}
\acmPrice{15.00}
\acmDOI{10.1145/3488560.3498382}
\acmISBN{978-1-4503-9132-0/22/02}

\begin{document}
\fancyhead{}

\title{Reconfiguration Problems on Submodular Functions}

\input{abstract}



\maketitle

\input{main}

\clearpage
\bibliographystyle{ACM-Reference-Format}
\bibliography{core,infmax}

\clearpage

\appendix%
\input{app}

\end{document}

%% file: macros.tex
\newtheorem{theorem}{Theorem}[section]
\newtheorem*{theorem*}{Theorem}
\newtheorem{lemma}[theorem]{Lemma}

\newtheorem{remark}[theorem]{Remark}

\newtheorem{problem}[theorem]{Problem}
\newtheorem{observation}[theorem]{Observation}
\newtheorem{definition}[theorem]{Definition}

\crefname{theorem}{Theorem}{Theorems}
\crefname{lemma}{Lemma}{Lemmas}
\crefname{claim}{Claim}{Claims}
\crefname{problem}{Problem}{Problems}
\crefname{remark}{Remark}{Remarks}
\crefname{observation}{Observation}{Observations}
\crefname{corollary}{Corollary}{Corollaries}
\crefname{appendix}{Appendix}{Appendices}
\crefformat{section}{#2\S#1#3}
\crefname{equation}{Eq.}{Eqs.}
\crefname{algorithm}{Algorithm}{Algorithms}
\crefname{figure}{Figure}{Figures}
\crefname{table}{Table}{Tables}

\usepackage{algorithmicx}
\usepackage[ruled]{algorithm}
\usepackage[noend]{algpseudocode}
\algnewcommand\And{\; \textbf{and} \;}
\algnewcommand\Or{\; \textbf{or} \;}
\algnewcommand\To{\; \textbf{to} \;}
\algnewcommand\Continue{\textbf{continue}}
\algnewcommand\Not{\textbf{not}}

\algnewcommand{\algalign}[1]{\parbox[t]{\dimexpr\linewidth-\algorithmicindent}{#1\strut}}
\algrenewcommand\textproc{\textsl}

\newcommand{\rme}{\mathrm{e}}

\newcommand{\calS}{\mathcal{S}}

\newcommand{\bbN}{\mathbb{N}}

\newcommand{\bbR}{\mathbb{R}}

\newcommand{\cP}{\textup{\textbf{P}}\xspace}
\newcommand{\NP}{\textup{\textbf{NP}}\xspace}
\newcommand{\shP}{\textbf{$\sharp$P}\xspace}
\newcommand{\PSPACE}{\textup{\textbf{PSPACE}}\xspace}
\newcommand{\NPSPACE}{\textup{\textbf{NPSPACE}}\xspace}

\DeclareMathOperator*{\argmax}{arg\,max}\renewcommand{\vec}[1]{\mathbf{\bm{#1}}}
\DeclareMathOperator{\Inf}{\mathsf{Inf}}

\newcommand{\mat}[1]{\mathbf{\bm{#1}}}
\newcommand{\prb}[1]{\textcolor{Firebrick4}{\textup{\textbf{\textsf{#1}}}}\xspace}

\newcommand{\bigO}{\mathcal{O}}
\newcommand{\USM}{\prb{Unconstrained Submodular Maximization}}
\newcommand{\MSM}{\prb{Monotone Submodular Maximization}}
\newcommand{\MSR}{\prb{MSReco}}
\newcommand{\USR}{\prb{USReco}}
\newcommand{\USRTAR}{\prb{USReco\textsf{[tar]}}}
\newcommand{\USRTJAR}{\prb{USReco\textsf{[tjar]}}}
\newcommand{\MMSR}{\prb{MaxMSReco}}
\newcommand{\MUSR}{\prb{MaxUSReco}}
\newcommand{\MUSRTAR}{\prb{MaxUSReco\textsf{[tar]}}}
\newcommand{\MUSRTJAR}{\prb{MaxUSReco\textsf{[tjar]}}}
\newcommand{\SATR}{\prb{3-SAT Reconfiguration}}
\newcommand{\MVCR}{\prb{Minimum Vertex Cover Reconfiguration}}

\newcommand{\TJ}{\textsf{tj}\xspace}
\newcommand{\TAR}{\textsf{tar}\xspace}
\newcommand{\TJAR}{\textsf{tjar}\xspace}

\newcommand{\OPT}{\mathsf{OPT}}
\newcommand{\True}{\textsf{True}}
\newcommand{\False}{\textsf{False}}

\newcommand{\OPEN}{\mathtt{OPEN}}
\newcommand{\CLOSE}{\mathtt{CLOSE}}

\newcommand{\karate}{\texttt{karate}\xspace}
\newcommand{\physicians}{\texttt{physicians}\xspace}

%% file: abstract.tex
\begin{abstract}%
\emph{Reconfiguration problems} require finding a step-by-step transformation between
a pair of feasible solutions for a particular problem.
The primary concern in Theoretical Computer Science has been
revealing their computational complexity for \emph{classical} problems.

This paper presents an initial study on reconfiguration problems derived from a \emph{submodular function}, which has more of a flavor of Data Mining.
Our submodular reconfiguration problems request to
find a solution sequence connecting two input solutions such that
each solution has an objective value above a threshold in a submodular function $f: 2^{[n]} \to \bbR_+$ and
is obtained from the previous one by applying a simple transformation rule.
We formulate three reconfiguration problems:
\prb{Monotone Submodular Reconfiguration} (\prb{MSReco}), which applies to influence maximization, and
two versions of \prb{Unconstrained Submodular Reconfiguration} (\prb{USReco}), which apply to determinantal point processes.
Our contributions are summarized as follows:
\begin{itemize}
    \item We prove that \prb{MSReco} and \prb{USReco} are both \PSPACE-complete.
    \item We design a $\frac{1}{2}$-approximation algorithm for \prb{MSReco} and a $\frac{1}{n}$-approximation algorithm for (one version of) \prb{USReco}.
    \item We devise inapproximability results that
    approximating the optimum value of \prb{MSReco} within a $(1-\frac{1+\epsilon}{n^2})$-factor is \PSPACE-hard, and
    we cannot find a $(\frac{5}{6}+\epsilon)$-approximation for \prb{USReco}.
    \item We conduct numerical study on the reconfiguration version of
    influence maximization and determinantal point processes using
    real-world social network and movie rating data.
\end{itemize}
\end{abstract}

%% file: main.tex
\section{Introduction}
\label{sec:intro}
Consider the following problem over the solution space:%
\begin{center}
\emph{Given a pair of feasible solutions for a particular source problem,\\ can we find a step-by-step transformation between them?}
\end{center}%
Such problems that involve transformation and movement are known by the name of \emph{reconfiguration problems} in Theoretical Computer Science \cite{heuvel13complexity,nishimura2018introduction,ito2011complexity}.
A famous example of reconfiguration problems is the 15 puzzle \cite{johnson1879notes}, where a \emph{feasible solution} is
an arrangement of $15$ numbered tiles on a $4 \times 4$ grid with one empty square, and a \emph{transformation} involves sliding a single tile to the empty square.
The goal is to transform from a given initial arrangement to the target arrangement such that the tiles are placed in numerical order.
This paper aims to introduce the concept of reconfiguration into Data Mining,
enabling us to \emph{connect} or \emph{interpolate} between
a pair of feasible solutions.
We explain two motivating examples of reconfiguration below:

\input{summary}

\input{fig-karate-vis}

\textbf{Influence Maximization Reconfiguration} (\cref{subsec:infmax}):
Suppose we are going to plan a \emph{viral marketing} campaign \cite{domingos2001mining} for promoting a company's new product.
Given structural data about a social network,
we can solve \emph{influence maximization} \cite{kempe2003maximizing}
to identify a small group of influential users.
However, the power of influence may decay as time goes by because social networks are \emph{evolving} \cite{leskovec2007graph,ohsaka2016dynamic} or
users may be affected by \emph{overexposure} \cite{loukides2020overexposure}.
One strategy to circumvent this issue is replacing an outdated group with a newly-found one.
When a change in user groups incurs a cost and we are given a limited budget (e.g., per day),
we need to \emph{interpolate} between an outdated group and a new one without significantly sacrificing the influence,
which entails the concept of reconfiguration.
\cref{fig:karate-vis} depicts an example of influence maximization.

\textbf{MAP Inference Reconfiguration} (\cref{subsec:map}):
Consider that we are required to arrange a list of items to be displayed on a recommender system.
If a feature vector is given for each item,
we can use a \emph{determinantal point process} \cite{borodin2005eynard,macchi1975coincidence} to extract a few items
achieving a good balance between item quality and set diversity \cite{kulesza2012determinantal,gillenwater2012near}.
Since \emph{novelty} plays a crucial role in increasing the recommendation utility \cite{vargas2011rank},
we would want to update the item list continuously.
On the other hand, we need to ensure \emph{stability} \cite{adomavicius2012stability}; i.e., the list should not be drastically changed over time,
which gives rise to reconfiguration.

Source problems for both examples are formulated as 
\prb{Submodular Maximization} \cite{nemhauser1978analysis,buchbinder2018deterministic,buchbinder2018submodular,krause2014submodular,buchbinder2015tight},
which finds many applications in data mining (see \cref{sec:related}).
Unfortunately,
the primary concern in the area of reconfiguration has been revealing
the computational complexity of reconfiguration problems for
\emph{classical} problems such as graph-algorithmic problems and Boolean satisfiability (see \cref{sec:related}),
which are incompatible with Data Mining and not applicable to the above examples.
Our objective is to formulate, analyze, and apply
reconfiguration problems derived from a \emph{submodular function}.

\subsection{Our Contributions}
We present an initial, systematic study on reconfiguration problems on submodular functions.
Our submodular reconfiguration problems request to determine
whether there exists a solution sequence connecting two input solutions such that 
each solution has an objective value above a threshold in a submodular function $f: 2^{[n]} \to \bbR_+$ and
is obtained from the previous one by applying a simple transformation rule (e.g., a single element addition and removal).
We formulate three submodular reconfiguration problems according to
\citet{ito2011complexity}'s framework of reconfiguration:
\begin{itemize}
\item \prb{Monotone Submodular Reconfiguration} (\MSR; \cref{prb:msr}):
This problem derives from \prb{Monotone Submodular Maximization}
and applies to influence maximization reconfiguration.
\item \prb{Unconstrained Submodular Reconfiguration} (\USRTAR and \USRTJAR; \cref{prb:usr-tar,prb:usr-tjar}):
These problems derive from \prb{Unconstrained Submodular Maximization} and two transformation rules.
MAP inference reconfiguration fits into them.
\end{itemize}
We further formulate the optimization variants (\cref{prb:mmmsr,prb:mmusr-tar,prb:mmusr-tjar}),
which aim to \emph{maximize} the minimum objective value among the solutions in the output sequence.
We analyze the proposed reconfiguration problems
through the lens of computational complexity.
Our complexity-theoretic results are summarized in \cref{tab:summary}.

\subsubsection*{\textbf{Hardness} (\cref{sec:hardness})}
We first investigate the computational tractability of
the submodular reconfiguration problems.
We prove that \MSR, \USRTAR, and \USRTJAR are all \PSPACE-complete (\cref{thm:msr-pspace,thm:usr-tar-pspace,thm:usr-tjar-pspace}),
which is at least as hard as \NP-completeness.

\subsubsection*{\textbf{Approximability} (\cref{sec:approx})}
Having established the hardness of solving
\MSR and \USR exactly,
we seek for \emph{approximation}
in terms of the minimum function value in the output sequence; namely,
we would like to maximize the minimum function value among the solutions as much as possible.
We design a $\max\{\frac{1}{2}, (1-\kappa)^2\}$-approximation algorithm for \MSR (\cref{thm:msr-approx}),
where $\kappa$ is the total curvature of a submodular function, and
a $\frac{1}{n}$-approximation algorithm for \USRTJAR (\cref{thm:usr-tjar-approx}).

\subsubsection*{\textbf{Inapproximability} (\cref{sec:inapprox})}
We further devise two \emph{hardness of approximation} results.
One is that
approximating the optimum value of \MSR within a factor of $(1-\frac{1+\epsilon}{n^2})$ is \PSPACE-hard (\cref{thm:msr-inapprox}),
implying that a fully polynomial-time approximation scheme does not exist assuming \cP~$\neq$~\PSPACE.
The other is that 
we cannot find a $(\frac{5}{6}+\epsilon)$-approximation 
for \USR, without making a complexity-theoretic assumption.
(\cref{thm:usr-tar-inapprox,thm:usr-tjar-inapprox}).

\subsubsection*{\textbf{Numerical Study} (\cref{sec:exp})}
We finally report numerical study on
the reconfiguration version of influence maximization \cite{kempe2003maximizing} using network data and
that of MAP inference on determinantal point process \cite{gillenwater2012near} using movie rating data,
which are formulated as \MSR and \USRTJAR, respectively.
Comparing to an A* search algorithm,
we observe that 
an approximation algorithm for \MSR quickly finds sequences
that are better than the worst-case analysis,
while that for \USRTJAR is far worse than the optimal sequence.

\section{Related Work}
\label{sec:related}

\subsubsection*{\textbf{Reconfiguration Problems}}
The concept of reconfiguration has arisen in problems involving transformation and movement, such as
the 15-puzzle \cite{johnson1879notes} and the Rubik's Cube.
\citet{ito2011complexity} established the unified framework of reconfiguration.
One of the most important reconfiguration problems is \emph{reachability},
asking to decide the existence of a solution sequence between two feasible solutions for a particular \emph{source} problem.
Countless source problems derive the respective reconfiguration problems in \citet{ito2011complexity}'s framework,
including graph-algorithmic problems, Boolean satisfiability, and others;
revealing their computational complexity has been the primary concern in Theoretical Computer Science.
Typically,
an \NP-complete source problem brings a \PSPACE-complete reachability problem, e.g., 
\prb{Vertex Cover} \cite{kaminski2012complexity},
\prb{Set Cover} \cite{ito2011complexity},
\prb{4-Coloring} \cite{bonsma2009finding},
\prb{Clique} \cite{ito2011complexity}, and
\prb{3-SAT} \cite{gopalan2009connectivity}.
On the other hand,
a source problem in \cP usually induces a reachability problem in \cP, e.g.,
\prb{Matching} \cite{ito2011complexity} and
\prb{2-SAT} \cite{gopalan2009connectivity}.
However, some exceptions are known; e.g.,
\prb{3-Coloring} is \NP-complete, but its reachability version is in \cP \cite{johnson2016finding}.
See \citet{nishimura2018introduction}'s survey  for more information.
This study explores reconfiguration problems for which the source problem is \prb{Submodular Maximization}, which generalizes \prb{Vertex Cover} and \prb{Set Cover} and has more of a flavor of Data Mining.

\subsubsection*{\textbf{Submodular Function Maximization}}
We review two submodular function maximization problems,
which have been studied in Theoretical Computer Science and applied in Data Mining.

Given a monotone submodular function,
\MSM requires finding 
a fixed-size set having the maximum function value.
The simple greedy algorithm has a provable guarantee of returning a $(1-1/\rme)$-factor approximation in polynomial time \cite{nemhauser1978analysis}.
This factor is the best possible as no polynomial-time algorithm can achieve a better approximation factor \cite{nemhauser1978best,feige1998threshold}.
Since monotone submodular functions abide by the law of \emph{diminishing returns},
\MSM has been applied to a diverse range of data mining tasks, e.g.,
influence maximization \cite{kempe2003maximizing}
document summarization \cite{lin2011class},
outbreak detection \cite{leskovec2007cost}, and
sensor placement \cite{krause2008near}.
We develop a $\frac{1}{2}$-approximation algorithm for the corresponding reconfiguration problem (\cref{subsec:msr-approx}).

Given a (not necessarily monotone) submodular function, 
\USM requires finding a subset that maximizes the function value.
This problem can be approximated within a $\frac{1}{2}$-factor \cite{buchbinder2015tight,buchbinder2018deterministic}, which
is proven to be optimal \cite{feige2011maximizing}.
Some of the application tasks include
movie recommendation, image summarization
\cite{mirzasoleiman2016fast},
and
MAP inference on determinantal point process \cite{gillenwater2012near}.
We develop a $\frac{1}{n}$-approximation algorithm for the corresponding reconfiguration problem (\cref{subsec:usr-tjar-approx}).

\section{Problem Formulation}
\label{sec:problem}

\subsubsection*{\textbf{Preliminaries}}
For a nonnegative integer $n$, let $[n] \triangleq \{1,2,\ldots, n\}$.
$\bbR_+$ represents the set of nonnegative real numbers.
For a finite set $S$ and a nonnegative integer $k$,
we write $ {S \choose k} $ for the family of all size-$k$ subsets of $S$.
A \emph{sequence} $\calS$ consisting of a finite number of sets $S^{(0)}, S^{(1)}, \ldots, S^{(\ell)}$
is denoted as $\langle S^{(0)}, S^{(1)}, \ldots, S^{(\ell)} \rangle$, and we write $S^{(i)} \in \calS$ to mean that $S^{(i)}$ appears in $\calS$ (at least once).
The symbol $\uplus$ is used to emphasize that the union is taken over two \emph{disjoint} sets.
Throughout this paper, we assume that every set function is nonnegative.
For a set function $f: 2^{[n]} \to \bbR_+$,
we say that
$f$ is \emph{monotone} if $f(S) \leq f(T)$ for all $S \subseteq T \subseteq [n]$,
$f$ is \emph{modular} if $f(S) + f(T) = f(S \cap T) + f(S \cup T)$ for all $S, T \subseteq [n]$,
and
$f$ is \emph{submodular} if
$f(S) + f(T) \geq f(S \cap T) + f(S \cup T)$ for all $S, T \subseteq [n]$.
Submodularity is known to be equivalent to the following \emph{diminishing returns} property \cite{schrijver2003combinatorial}:
$f(S \cup \{e\}) - f(S) \geq f(T \cup \{e\}) - f(T)$ for all
$S \subseteq T \subseteq [n]$ and $e \in [n] \setminus T$.
For a subset $R \subseteq [n]$,
the \emph{residual} \cite{krause2014submodular} is defined as a set function $f_R: 2^{[n] \setminus R} \to \bbR_+$ such that
$f_R(S) \triangleq f(S \uplus R) - f(R)$ for $S \subseteq [n] \setminus R$.
If $f$ is monotone and submodular, then so is $f_R$ \cite{krause2014submodular}.
The \emph{total curvature} $\kappa$ \cite{conforti1984submodular,vondrak2010submodularity}
of a monotone submodular function $f: 2^{[n]} \to \bbR_+$ is defined as 
$
    \kappa \triangleq 1 - \min_{e \in [n]}\frac{f([n]) - f([n] \setminus \{e\})}{f(\{e\})}.
$
The total curvature $\kappa$ takes a value from $0$ to $1$, which captures how far away $f$ is from being modular; e.g.,
a modular function has curvature $\kappa=0$, and 
a coverage function has curvature $\kappa=1$.\footnote{
Given a collection of $n$ subsets $A_1, \ldots, A_n$ of some ground set $U$, we refer to a set function $f: 2^{[n]} \to \bbR_+$ such that $f(S) \triangleq |\bigcup_{i \in S}A_i|$ as a \emph{coverage function}.}
We assume to be given access to a \emph{value oracle} for a set function $f$, which returns $f(S)$ whenever it is called with a query $S$.
We recall the definitions of two submodular function maximization problems:

\begin{itemize}
    \item[\textbf{1.}] \MSM: Given a monotone submodular function $f: 2^{[n]} \to \bbR_+$ and a solution size $k$, 
    maximize $f(S)$ subject to $S \in {[n] \choose k}$.
    \item[\textbf{2.}] \USM: Given a submodular function $f: 2^{[n]} \to \bbR_+$,
    maximize $f(S)$ subject to $S \subseteq [n]$.
\end{itemize}

\subsection{\citet{ito2011complexity}'s Reconfiguration Framework}
In reconfiguration problems, we wish to determine whether there exists a sequence of solutions between a pair of solutions for a particular ``source'' problem such that each is ``feasible'' and obtained from the previous one by applying a simple ``transformation rule.''
We recapitulate the reconfiguration framework of \citet{ito2011complexity}.
The reconfiguration framework requires three ingredients \cite{nishimura2018introduction,ito2011complexity,mouawad2015reconfiguration}:
\begin{itemize}
\item[\textbf{1.}]\textbf{a source problem}, which is usually a search problem in \cP or \NP-complete;
\item[\textbf{2.}]\textbf{a definition of feasible solutions};
\item[\textbf{3.}]\textbf{an adjacency relation} over the pairs of two solutions, typically symmetric and polynomial-time testable \cite{nishimura2018introduction}.
\end{itemize}
An adjacency relation can be defined in terms of
a \emph{reconfiguration step}, which specifies how a solution can be transformed. We say that two solutions are \emph{adjacent} if one can be transformed into the other by applying a single reconfiguration step.
We now define a central concept called \emph{reconfiguration sequences}.

\begin{definition}
For two feasible solutions $X$ and $Y$,
a \emph{reconfiguration sequence from $X$ to $Y$}
is a sequence of feasible solutions
$\calS = \langle S^{(0)}, S^{(1)}, \ldots, S^{(\ell)} \rangle $
starting from $X$ (i.e., $S^{(0)}=X$) and ending with $Y$ (i.e., $S^{(\ell)} = Y$) such that
every two consecutive solutions $S^{(i-1)}$ and $S^{(i)}$ for $i \in [\ell]$ are adjacent
(i.e., $S^{(i)}$ is obtained from $S^{(i-1)}$ by a single reconfiguration step).
The \emph{length} $\ell$ of $\calS$ is defined as the number of (possibly duplicate) solutions in it minus $1$.
\end{definition}

There are several types of reconfiguration problems \cite{nishimura2018introduction,mouawad2015reconfiguration,heuvel13complexity}.
One of the most important problems is \emph{reachability},
asking to determine whether there exists a reconfiguration sequence between a pair of feasible solutions.
Of course, reconfiguration problems for the same source problem can have
different complexities depending on the definitions of feasibility and adjacency.

\subsection{Defining Submodular Reconfiguration}
\label{subsec:defining}
We are now ready to formulate reconfiguration problems on submodular functions.
We first designate the ingredients required for defining the reconfiguration framework.
Source problems are either \MSM or \USM.
Given a submodular function $f:2^{[n]} \to \bbR_+$,
we define the feasibility according to \cite[\S 2.2]{ito2011complexity}:
We introduce a \emph{threshold} $\theta$, offering a lower bound on the allowed function values, and a set $S$ is said to be \emph{feasible} if $f(S) \geq \theta$.
For a set sequence $\calS$,
the \emph{value} of $\calS$, denoted $f(\calS)$, is defined as
the minimum function value among all sets in $\calS$, i.e.,
$f(\calS) \triangleq \min_{S^{(i)} \in \calS} f(S^{(i)})$.
Accordingly, a reconfiguration sequence $\calS$ must satisfy $f(\calS) \geq \theta$.
We consider three reconfiguration steps to specify an adjacency relation,
some of which are established in the literature:
\begin{itemize}
    \item[\textbf{1.}]\emph{Token jumping} (\TJ) \cite{kaminski2012complexity}:
    Given a set,
    a \TJ step can remove one element from it \emph{and} add another element not in it at the same time; i.e.,
    two sets are adjacent under \TJ if they have the same size and their intersection has a size one less than their size.
    \item[\textbf{2.}]\emph{Token addition or removal} (\TAR) \cite{ito2011complexity}:
    Given a set, a \TAR step can 
    remove an element from it \emph{or} add an element not in it; i.e.,
    two sets are adjacent under \TAR if the symmetric difference has size $1$.
    \item[\textbf{3.}]\emph{Token jumping, addition, or removal} (\TJAR): A \TJAR step can perform either a \TJ or \TAR step.
\end{itemize}
It is easy to see that these adjacency relations are symmetric and polynomial-time testable.

\subsubsection{\textbf{Reachability Problems}}
We define three reachability problems on a submodular function with different adjacency relations.\footnote{We do not consider \MSR under \TAR or \TJAR since they yield a set not in ${n \choose k}$.}

\begin{problem}[\prb{Monotone Submodular Reconfiguration}; \MSR]
\label{prb:msr}
Given a monotone submodular function $f: 2^{[n]} \to \bbR_+$,
two sets $X$ and $Y$ in ${[n] \choose k}$, and
a threshold $\theta$,
decide if there exists a reconfiguration sequence $\calS$ from $X$ to $Y$ under \TJ such that
$f(\calS) \geq \theta$.
\end{problem}

\begin{problem}[\prb{Unconstrained Submodular Reconfiguration in \TAR}; \USRTAR]
\label{prb:usr-tar}
Given a submodular function $f: 2^{[n]} \to \bbR_+$,
two subsets $X$ and $Y$ of $[n]$, and
a threshold $\theta$,
decide if there exists a reconfiguration sequence $\calS$ from $X$ to $Y$ under \TAR such that
$f(\calS) \geq \theta$.
\end{problem}

\begin{problem}[\prb{Unconstrained Submodular Reconfiguration in \TJAR}; \USRTJAR]
\label{prb:usr-tjar}
Given a submodular function $f: 2^{[n]} \to \bbR_+$,
two subsets $X$ and $Y$ of $[n]$, and
a threshold $\theta$,
decide if there exists a reconfiguration sequence $\calS$ from $X$ to $Y$ under \TJAR such that
$f(\calS) \geq \theta$.
\end{problem}

Note that these problems do not request an actual reconfiguration sequence.
Without loss of generality,
we assume that $\theta$ is at most $\min\{ f(X), f(Y) \}$,
because otherwise the answer is always ``no.''

\subsubsection{\textbf{Optimization Variants}}
By definition,
the answer to \cref{prb:msr,prb:usr-tar,prb:usr-tjar} is always ``yes''
if $\theta = 0$.
On the other hand, there exists a constant $\theta_{\mathrm{yes}}$,
referred to as a \emph{reconfiguration index} \cite{ito2016reconfiguration},
for which the answer is ``yes'' if $\theta \leq \theta_{\mathrm{yes}}$ and ``no'' otherwise.
We can thus think of the following \emph{optimization variants},
requiring that $f(\calS)$ be maximized among all possible reconfiguration sequences.
Such variants have been studied for \prb{Clique} \cite{ito2011complexity} and \prb{Subset Sum} \cite{ito2014approximability}.

\begin{problem}[\prb{Maximum Monotone Submodular Reconfiguration}; \MMSR]
\label{prb:mmmsr}
Given a monotone submodular function $f: 2^{[n]} \to \bbR_+$ and
two sets $X$ and $Y$ in ${[n] \choose k}$,
find a reconfiguration sequence $\calS$ from $X$ to $Y$
under \TJ
maximizing $f(\calS)$.
\end{problem}

\begin{problem}[\prb{Maximum Unconstrained Submodular Reconfiguration in \TAR}; \MUSRTAR]
\label{prb:mmusr-tar}
Given a submodular function $f: 2^{[n]} \to \bbR_+$ and
two subsets $X$ and $Y$ of $[n]$,
find a reconfiguration sequence $\calS$ from $X$ to $Y$
under \TAR maximizing $f(\calS)$.
\end{problem}

\begin{problem}[\prb{Maximum Unconstrained Submodular Reconfiguration in \TJAR}; \MUSRTJAR]
\label{prb:mmusr-tjar}
Given a submodular function $f: 2^{[n]} \to \bbR_+$ and
two subsets $X$ and $Y$ of $[n]$,
find a reconfiguration sequence $\calS$ from $X$ to $Y$
under \TJAR maximizing $f(\calS)$.
\end{problem}

\section{Hardness}
\label{sec:hardness}

In this section, we prove that \MSR, \USRTAR, and \USRTJAR are all \PSPACE-complete to solve (\cref{thm:msr-pspace,thm:usr-tar-pspace,thm:usr-tjar-pspace}).
Here, \PSPACE is a class of decision problems that can be solved using \emph{polynomial space} in the input size, and a decision problem is said to be \emph{\PSPACE-complete} if it is in \PSPACE and every problem in \PSPACE can be reduced to it in polynomial time.
\PSPACE is known to include (and believed to be outside \cite{arora2009computational}) \cP, \NP, and \shP.
Commonly known \PSPACE-complete problems are \prb{Quantified Boolean Formula} \cite{garey1979computers}, puzzles and games such as \prb{Sliding Blocks} \cite{hearn2005pspace} and \prb{Go} \cite{lichtenstein1980go}.
We can easily verify that submodular reconfiguration problems are included in \PSPACE, whose proof is deferred to \cref{app:proofs}.

\begin{observation}
\label{obs:pspace}
\cref{prb:msr,,prb:usr-tar,,prb:usr-tjar} are in \PSPACE.
\end{observation}

\subsection{\PSPACE-completeness of \MSR}
\begin{theorem}
\label{thm:msr-pspace}
\MSR is \PSPACE-complete.
\end{theorem}

To prove \cref{thm:msr-pspace}, we use a polynomial-time reduction from \MVCR.
Of a graph, a \emph{vertex cover} is a set of vertices that
include at least one endpoint of every edge of the graph.
Given a graph and an integer $k$,
it is \NP-complete to decide if
there exists a vertex cover of size $k$ \cite{karp1972reducibility}.
We define \MVCR as follows.

\begin{problem}[\MVCR]
\label{prb:rvcr}
Given a graph $G = (V,E)$ and
two minimum vertex covers $C^x$ and $C^y$ of the same size,
determine whether there exists a sequence of minimum vertex covers
from $C^x$ to $C^y$ under \TJ.
\end{problem}

Our definition is different from that of \prb{Vertex Cover Reconfiguration} due to \cite{ito2011complexity,ito2016reconfiguration}, in which two input vertex covers may not be minimum.
We show that \cref{prb:rvcr} is \PSPACE-hard,
whose proof is reminiscent of \cite[Theorem 2]{ito2011complexity} and deferred to \cref{app:proofs}.

\begin{lemma}
\label{lem:rvcr-pspace}
\cref{prb:rvcr} is \PSPACE-hard.
\end{lemma}
\begin{proof}[Proof of \cref{thm:msr-pspace}]

We present a polynomial-time reduction from \MVCR.
Suppose we are given
a graph $G=(V,E)$ and two minimum vertex covers $C^x$ and $C^y$.
Define a set function $f: 2^{V} \to \bbR_+$ such that
$f(S)$ for $S \subseteq V$ is the number of edges in $E$ that are incident to $S$.
In particular, $f(S) = |E|$ if and only if $S$ is a vertex cover of $G$.
Since $f$ is monotone and submodular,
we construct an instance of \MSR consisting of
$f$, $C^x$, $C^y$, and
a threshold $|E|$.
Observe that
a reconfiguration sequence for the \MVCR instance is
a reconfiguration sequence for the \MSR instance, and vice versa,
which completes the reduction.
\end{proof}

\subsection{\PSPACE-completeness of \USR}
\begin{theorem}
\label{thm:usr-tar-pspace}
\USRTAR is \PSPACE-complete.
\end{theorem}

\begin{proof}
We demonstrate a polynomial-time reduction from
\prb{Monotone Not-All-Equal 3-SAT Reconfiguration},
which is \PSPACE-complete \cite{cardinal2020reconfiguration}.
A \emph{3-conjunctive normal form (3-CNF) formula} $\phi$ is said to be \emph{monotone} if no clause contains negative literals
(e.g., $\phi = (x_1 \vee x_2 \vee x_3) \wedge (x_2 \vee x_3 \vee x_4) \wedge (x_1 \vee x_3 \vee x_4)$).
We assume that every clause of $\phi$ contains exactly three literals.
We say that a truth assignment $\vec{\sigma}$
\emph{not-all-equal satisfies} $\phi$ if
every clause contains exactly two literals with the same value; i.e.,
it contains at least one true literal and at least one false literal (e.g., $\vec{\sigma}(x_1) = \vec{\sigma}(x_2) = \vec{\sigma}(x_4) = \True$ and $\vec{\sigma}(x_3) = \False$).
In \prb{Monotone Not-All-Equal 3-SAT Reconfiguration},
given a monotone 3-CNF formula $\phi$ and
two not-all-equal satisfying truth assignments $\vec{\sigma}^x$ and $\vec{\sigma}^y$ of $\phi$,
we wish to determine whether there exists a sequence of not-all-equal satisfying truth assignments of $\phi$ between
$\vec{\sigma}^x$ and $\vec{\sigma}^y$
such that each truth assignment is obtained from the previous one by a single variable flip; i.e., they differ in exactly one variable (cf.~\SATR \cite{gopalan2009connectivity} in \cref{prb:satr}).

Suppose we are given a monotone 3-CNF formula $\phi$ with $n$ variables $x_1, \ldots, x_n$ and $m$ clauses $c_1, \ldots, c_m$ and
two not-all-equal satisfying truth assignments $\vec{\sigma}^x$ and $\vec{\sigma}^y$ of $\phi$.
For a subset $S \subseteq [n]$,
we write $\vec{\sigma}_S$ for a truth assignment such that
$\vec{\sigma}_S(x_i)$ for variable $x_i$ is $\True$ if $i \in S$ and $\False$ otherwise.
For a truth assignment $\vec{\sigma}$,
we define the set $S_{\vec{\sigma}} \triangleq \{ i \in [n] \mid \vec{\sigma}(x_i) = \True \}$.
We now construct a set function $f_\phi: 2^{[n]} \to \bbR_+$ such that
$f(S)$ for $S \subseteq [n]$ is
the number of clauses not-all-equal satisfied by $\vec{\sigma}_S$.
In particular, $f(S) = m$ if $\vec{\sigma}_S$
not-all-equal satisfies $\phi$.
Since $f$ is submodular,\footnote{
Suppose $\phi$ contains a single clause, say, $\phi = x_1 \vee x_2 \vee x_3$.
Then, $f_\phi$ can be written as
$f_\phi(S) = g(|S \cap [3]|)$,
where $g: \bbN \to \bbR_+$ is defined as $g(0)=g(3)=0$ and $g(1)=g(2)=1$.
Since $g$ is concave, $f_\phi$ is submodular \cite[Proposition 5.1]{lovasz1983submodular}.
}
we construct an instance of \USR consisting of
$f$,
$S_{\vec{\sigma}^x}$, $S_{\vec{\sigma}^y}$, and
a threshold $m$.
Observe that
there exists a reconfiguration sequence for the \prb{Monotone Not-All-Equal 3-SAT Reconfiguration} instance
\emph{if and only if}
there exists a reconfiguration sequence for the \USR instance,
which completes the reduction.
\end{proof}

The last \PSPACE-completeness result is shown below,
whose proof is based on a reduction from \MVCR and deferred to \cref{app:proofs}.

\begin{theorem}
\label{thm:usr-tjar-pspace}
\USRTJAR is \PSPACE-complete.
\end{theorem}

\input{alg-greedy}

\section{Approximability}
\label{sec:approx}

In the previous section, we saw that 
\MSR, \USRTAR, and \USRTJAR are all \PSPACE-complete,
implying that their optimization variants
are also hard to solve exactly in polynomial time.
However, there is still room for consideration of \emph{approximability}.
A \emph{$\rho$-approximation algorithm} for $\rho \leq 1$ is 
a polynomial-time algorithm that returns a reconfiguration sequence $\calS$ such that
$f(\calS) \geq \rho \cdot f(\calS^*)$,
where $\calS^*$ is an optimal reconfiguration sequence with the maximum value.
We design a $\max\{\frac{1}{2}, (1-\kappa)^2\}$-approximation algorithm for \MMSR
(\cref{subsec:msr-approx}; \cref{thm:msr-approx}), where
$\kappa$ is the total curvature, and
a $\frac{1}{n}$-approximation algorithm for \MUSRTJAR (\cref{subsec:usr-tjar-approx}; \cref{thm:usr-tjar-approx}),
while we explain the difficulty in algorithm development for \MUSRTAR (\cref{subsec:usr-tar-approx}).

\subsection{Greedy Algorithm}
Before going into details of the proposed algorithms, we introduce the \emph{greedy algorithm} shown in \cref{alg:greedy}, which is used as a subroutine.
Given a set function $f: 2^{[n]} \to \bbR_+$,
a ground set $N \subseteq [n]$, and a solution size $k \leq |N|$,
the greedy algorithm iteratively selects
an element of $N$, not having been chosen so far, that maximizes the function value.
The number of calls to a value oracle of $f$ is at most $|N|k$.
Let $e_i$ denote an element chosen at the $i$-th iteration;
define $S_i \triangleq \{e_1, \ldots, e_i\}$.
We call the output sequence $\langle e_1, \ldots, e_k \rangle$ a \emph{greedy sequence}.
If $f$ is a submodular function, then
the following inequality is known to hold for any $1 \leq i \leq j \leq k$, see, e.g., \cite{krause2014submodular}:
\begin{align}
\label{eq:greedy-ineq}
    f(S_i) - f(S_{i-1}) \geq f(S_j) - f(S_{j-1}).
\end{align}

\subsection{$\max\{\frac{1}{2}, (1-\kappa)^2\}$-Approximation Algorithm for \MMSR}
\label{subsec:msr-approx}

\cref{alg:swap} describes the proposed approximation algorithm for \MMSR.
Given a monotone submodular function $f: 2^{[n]} \to \bbR_+$ and
two sets $X$ and $Y$ in ${[n] \choose k}$,
it first invokes \cref{alg:greedy} on
$f_R$,
$X \setminus R$, and $k'$
(resp.~$f_R$, $Y \setminus R$, and $k'$),
where $R \triangleq X \cap Y$ and $k' \triangleq |X \setminus R| = |Y \setminus R|$,
to obtain a greedy sequence
$\langle x_1, \ldots, x_{k'} \rangle$
(resp.~$\langle y_1, \ldots, y_{k'} \rangle$).
It then returns
a set sequence from $X$ to $Y$,
the $i$-th set in which is defined as
$S^{(i)} \triangleq \{x_1, \ldots, x_{k'-i}\} \uplus \{ y_1, \ldots, y_i \} \uplus R$.
Our algorithm is guaranteed to return
a $\max\{\frac{1}{2}, (1-\kappa)^2\}$-approximation reconfiguration sequence in $\bigO(nk)$ time, where $\kappa$ is the total curvature of $f$.
\cref{alg:swap} is thus nearly optimal whenever $\kappa \approx 0$.
Such a small $\kappa$ can be observed in real-world problems, e.g., entropy sampling on Gaussian radial basis function kernels \cite{sharma2015greedy}.
\input{alg-swap}

\begin{theorem}
\label{thm:msr-approx}
Given a monotone submodular function $f: 2^{[n]} \to \bbR_+$ and
two sets $X$ and $Y$ in ${[n] \choose k}$,
\cref{alg:swap} returns a reconfiguration sequence
$\calS$ for \MMSR of length at most $k$ in $\bigO(nk)$ time such that
$f(\calS) \geq \max\{\frac{1}{2}, (1-\kappa)^2\} \min\{f(X), f(Y)\}$.
In particular, it is a $\max\{\frac{1}{2}, (1-\kappa)^2\}$-approximation algorithm for \MMSR.
\end{theorem}

\begin{proof}
Define $R \triangleq X \cap Y$, 
$X' \triangleq X \setminus R$,
$Y' \triangleq Y \setminus R$, and
$k' \triangleq |X'| = |Y'| $.
Let $\langle x_1, \ldots, x_{k'} \rangle$
(resp.~$\langle y_1, \ldots, y_{k'} \rangle$)
denote the greedy sequence returned by
\cref{alg:greedy} invoked on $f_R, X', k'$ (resp.~$f_R, Y', k'$).
For each $i \in \{0\}\cup[k']$,
we define
$X_i \triangleq \{x_1, \ldots, x_i\}$ and
$Y_i \triangleq \{y_1, \ldots, y_i\}$.
Note that $X_0 = Y_0 = \emptyset$, $X_{k'} = X'$, and $Y_{k'} = Y'$.
For each $i \in \{0\}\cup [k']$,
$S^{(i)}$ in the returned reconfiguration sequence  $\calS$ is equal to $X_{k'-i} \uplus Y_i \uplus R$, which is of size $k$.
The correctness of \cref{alg:swap} comes from the fact that
$S^{(i)}$ is obtained from $S^{(i-1)}$
by removing $x_{k'-i+1}$ and adding $y_i$.
The time complexity is apparent.

Showing that
$f(S^{(i)}) \geq \frac{1}{2} \min\{f(X), f(Y)\}$ for
every $i$ now suffices to prove a $\frac{1}{2}$-approximation.
Since the statement is clear if $i = 0, k'$, we will prove for the case of $i \in [k'-1]$.
Observe first that, whenever $i \leq j$, we have that
$
    f_R(X_i) - f_R(X_{i-1}) \geq f_R(X_j) - f_R(X_{j-1})
$
due to \cref{eq:greedy-ineq}.
Hence, for any $i \in [k'-1]$,
we have that
\begin{align*}
    \tfrac{1}{i} {\textstyle \sum\limits_{1 \leq j \leq i}} (f_R(X_j) - f_R(X_{j-1}))
    \geq \tfrac{1}{k'-i} {\textstyle \sum\limits_{i+1 \leq j \leq k'}}
    (f_R(X_j) - f_R(X_{j-1})).
\end{align*}
Simple calculation further yields that
$f_R(X_i) \geq \frac{i}{k'} f_R(X')$,
where we have used the nonnegativity of $f_R(\emptyset)$.
Similarly, we can show that
$f_R(Y_i) \geq \frac{i}{k'} f_R(Y')$ for every $i \in [k'-1]$.
Using the two inequalities on $f_R(X_{k'-i})$ and $f_R(Y_i)$,
we have that for any $i \in [k'-1]$,
\begin{align}
& f(S^{(i)}) = f_R(X_{k'-i} \uplus Y_i) + f(R) \geq \max\{f_R(X_{k'-i}), f_R(Y_i)\} + f(R) \nonumber \\
& \geq \max\bigl\{ \tfrac{k'-i}{k'}f_R(X'), \tfrac{i}{k'}f_R(Y') \bigr\} + f(R) 
= \tfrac{1}{2} \min\{ f(X), f(Y) \}, \label{eq:mmsr-proof}
\end{align}
where the first inequality is due to the monotonicity of $f_R$.
Proving a $(1-\kappa)^2$-approximation is deferred to \cref{app:proofs}.
\end{proof}

\subsubsection*{\textbf{Difficult Instance for \cref{alg:swap}}}
We provide a specific instance of \MMSR for which \cref{alg:swap} returns a $\frac{3}{4}$-approximation reconfiguration sequence, whose proof is deferred to \cref{app:proofs}.
As a by-product, we give evidence that
an optimal reconfiguration sequence can include elements \emph{outside} $X \cup Y$.

\begin{observation}
\label{obs:msr-3-4}
There exists an instance $f, X, Y$ of  \MMSR
such that the optimal reconfiguration sequence $\calS^*$ has value $f(\calS^*) = 1$, and
any reconfiguration sequence $\calS$ that is restricted to include only subsets of $X \cup Y$ has value $f(\calS) \leq \frac{3}{4}$.
Thus, \cref{alg:swap} returns
a $\frac{3}{4}$-approximation reconfiguration sequence for this instance.
\end{observation}

\subsection{$\frac{1}{n}$-Approximation Algorithm for \MUSRTJAR}
\label{subsec:usr-tjar-approx}
\input{alg-usr}

\cref{alg:usr} describes
the proposed approximation algorithm for \MUSRTJAR.
Given a submodular function $f: 2^{[n]} \to \bbR_+$ and
two subsets $X$ and $Y$ of $[n]$,
it first invokes \cref{alg:greedy} on $f,X,|X|$ and $f,Y,|Y|$ to
obtain the greedy sequences $\langle x_1, \ldots, x_{|X|} \rangle$ and
$\langle y_1, \ldots, y_{|Y|} \rangle$,
respectively.
It then returns the concatenation of a reconfiguration sequence from $X$ to $\{x_1\}$ and 
that from $\{y_1\}$ to $Y$.
Our algorithm is guaranteed to return
a $\frac{1}{n}$-approximation reconfiguration sequence in $\bigO(n^2)$ time
as claimed below.

\begin{theorem}
\label{thm:usr-tjar-approx}
Given a submodular function $f: 2^{[n]} \to \bbR_+$ and two subsets $X$ and $Y$ of $[n] $,
\cref{alg:usr} returns a reconfiguration sequence $\calS$
for \MUSRTJAR of length at most $2n$ in $\bigO(n^2)$ time such that
$f(\calS) \geq \frac{1}{n} \min\{ f(X), f(Y) \}$.
In particular,
it is a $\frac{1}{n}$-approximation algorithm for \MUSRTJAR.
\end{theorem}
\begin{proof}
Let $\langle x_1, \ldots, x_{|X|} \rangle$
(resp.~$\langle y_1, \ldots, y_{|Y|} \rangle$)
denote the greedy sequence returned by
\cref{alg:greedy} invoked on $f,X,|X|$ (resp.~$f,Y,|Y|$).
For each $i \in \{0\} \cup [|X|]$ (resp.~$i \in \{0\} \cup [|Y|]$),
we define $X_i \triangleq \{x_1, \ldots, x_i\}$
(resp.~$Y_i \triangleq \{y_1, \ldots, y_i\}$).
Observe that the sequence $\calS$ returned by \cref{alg:usr} is a valid reconfiguration sequence for \MUSRTJAR and consists of sets in the form of either $X_i$ or $Y_i$ for some $i \geq 1$.
The time complexity is obvious.

We will show that $f(X_i) \geq \frac{1}{|X|}f(X)$ for every $i \geq 1$.
Since the value of $f(X_i) - f(X_{i-1})$
is monotonically nonincreasing in $i \in [|X|]$ owing to \cref{eq:greedy-ineq},
there exists an index $i^*$ such that
$f(X_i) - f(X_{i-1}) \geq 0$ if $i \leq i^*$ and
$f(X_i) - f(X_{i-1}) \leq 0$ if $i \geq i^*+1$.
In the former case, we have that
$
    f(X_i)
    \geq f(X_1);
$
in the latter case, we have that
$
    f(X_i)
    \geq f(X).
$
Using the inequality that
$f(X) \leq \sum_{i \in [|X|]} f(\{x_i\}) \leq |X| \cdot f(\{x_1\})$,
we obtain that 
$
    f(X_i) \geq \min\{ f(X_1), f(X) \} \geq \frac{1}{|X|} f(X)
$
for any $i \in [|X|]$.
Similarly, we can derive an analogous inequality that for any $i \in [|Y|]$,
$
    f(Y_i) \geq \frac{1}{|Y|} f(Y).
$
Accordingly, we derive that
\begin{align}
\label{eq:musr-proof}
    f(\calS) \geq \min\bigl\{ \tfrac{1}{|X|} f(X), \tfrac{1}{|Y|} f(Y) \bigr\} \geq \tfrac{1}{n} \min\{f(X), f(Y)\},
\end{align}
which completes the proof.
\end{proof}

\subsubsection*{\textbf{Does \cref{alg:swap} Work on \MUSRTJAR?}}
\cref{alg:usr}'s approximation factor of $\frac{1}{n}$
is not fascinating compared to a $\frac{1}{2}$-factor of \cref{alg:swap} on \MMSR.
One might wonder if \cref{alg:swap} generates a good reconfiguration sequence \emph{on} \MUSRTJAR, assuming that $|X|=|Y|$.
However, we have bad news that \cref{alg:swap} does not have \emph{any} positive approximation factor for \MUSRTJAR,
whose proof is deferred to \cref{app:proofs}.

\begin{observation}
\label{obs:musrtjar}
There exists an instance $f,X,Y$ of \MUSRTJAR with $|X|=|Y|$ such that
\cref{alg:usr} and \cref{alg:swap} return a reconfiguration sequence of value $1$ and $0$, respectively.
\end{observation}

\subsection{Difficulty in Designing Approximation Algorithms for \MUSRTAR}
\label{subsec:usr-tar-approx}

Unfortunately, \cref{alg:usr} designed for \MUSRTJAR does not
produce a reconfiguration sequence for \MUSRTAR because
we cannot transform from $\{x_1\}$ to $\{y_1\}$ directly by a \TAR step.
Here, we explain what makes it so challenging to design approximation algorithms for \MUSRTAR.
\cref{eq:mmsr-proof,eq:musr-proof} in the proofs of \cref{thm:msr-approx,thm:usr-tjar-approx} indicate that
if $f(X)$ and $f(Y)$ are positive,
then there must exist a reconfiguration sequence $\calS$ whose value is positive (which can be found efficiently).
Such a feature is critical for proving
$f(\calS) \geq \rho \cdot \min\{f(X), f(Y)\}$ for some positive $\rho > 0$.
We show, however, that this is not the case for \MUSRTAR; i.e., it can be impossible to transform from $X$ to $Y$ without ever touching \emph{zero-value} sets,
whose proof is deferred to \cref{app:proofs}.

\begin{observation}
\label{obs:usr-tar-approx}
There exists an instance $f,X,Y$ of \MUSRTAR
such that $f(X) = f(Y) = 1$ and 
every reconfiguration sequence $\calS$ has value $f(\calS) = 0$.
\end{observation}

\section{Inapproximability}
\label{sec:inapprox}

In this section,
we devise inapproximability results of
\MMSR and \MUSR,
which reveal an \emph{upper bound} of approximation guarantees that polynomial-time algorithms can achieve.
We first prove that it is \PSPACE-hard to approximate the optimal value of \MMSR within a factor of $(1-\frac{1+\epsilon}{n^2})$ (\cref{subsec:msr-inapprox}; \cref{thm:msr-inapprox}), which is slightly stronger than \cref{thm:msr-pspace}.
Though this factor asymptotically approaches $1$ (as $n$ goes to infinity),
the result rules out the existence of a fully polynomial-time approximation scheme, assuming that \cP~$\neq$~\PSPACE (which is a weaker assumption than \cP~$\neq$~\NP).
A \emph{fully polynomial-time approximation scheme (FPTAS)} is 
an approximation algorithm that takes a precision parameter $\epsilon>0$ and returns
a $(1-\epsilon)$-approximation in polynomial time in the input size and $\epsilon^{-1}$.
We then show that both versions of \MUSR cannot be approximated within a factor of $(\frac{5}{6}+\epsilon)$ for any $\epsilon > 0$ by using exponentially many oracle calls in $n$ and $\epsilon$, without making a complexity-theoretic assumption (\cref{subsec:usr-inapprox}; \cref{thm:usr-tjar-inapprox,thm:usr-tar-inapprox}).

\subsection{Inapproximability Result of \MMSR}
\label{subsec:msr-inapprox}

The first result is shown below,
whose proof appears in \cref{app:proofs}.

\begin{theorem}
\label{thm:msr-inapprox}
It is \PSPACE-hard to approximate the optimal value of \MMSR
within a factor of $(1-\frac{1+\epsilon}{n^2})$ for any $\epsilon > 0$,
where $n$ is the size of the ground set.
In particular, an FPTAS for \MMSR does not exist unless
\cP $=$ \PSPACE.
\end{theorem}

\subsection{Inapproximability Results of \MUSR}
\label{subsec:usr-inapprox}

\begin{theorem}
\label{thm:usr-tjar-inapprox}
For any $\epsilon > 0$,
there is no $(\frac{5}{6}+\epsilon)$-approximation algorithm for \MUSRTJAR making at most $\rme^{\epsilon^2 n/2}$ oracle calls.
\end{theorem}
\begin{proof}
We show a reduction from \USM in an approximation-preserving manner.
Suppose we are given a submodular function
$f: 2^{[n]} \to \bbR_+$ and a number $\epsilon > 0$, and we wish to find
a $(\frac{1}{2} + \epsilon)$-approximation for \USM.
We first compute a $\frac{1}{2}$-approximation $\hat{O}$ in polynomial time \cite{buchbinder2018deterministic}, and
we define $\Upsilon \triangleq (2+2 \epsilon) f(\hat{O}) $.
We have that $(1+\epsilon) \OPT \leq \Upsilon \leq (2+2\epsilon) \OPT$,
where $\OPT \triangleq \max_{S \subseteq [n]} f(S)$ (which is unknown).
We can safely assume that $\Upsilon > 0$ because otherwise we can declare that the optimal value is $\OPT = 0$.
Define $V \triangleq \{x_1,x_2,y_1,y_2\}$ and $N \triangleq [n] \uplus V$.
We then construct a submodular function $g: 2^{N} \to \bbR_+$ such that
$g(T) \triangleq \frac{\Upsilon}{2} \cdot c(T \cap V)  + f(T \cap [n])$ for each $T \subseteq N$, where
$c$ is a \emph{cut function} on graph $G=(V,E)$ with $E = \{(x_1,y_1), (x_1,y_2), (x_2,y_1), (x_2,y_2)\}$.
Since $\frac{\Upsilon}{2} \cdot c(\cdot)$ takes either of $0, \Upsilon, 2 \Upsilon$ as a value and
$f(\cdot)$ takes a value within the range of $[0, \OPT]$,
we have the following relation between $g(T)$ and $c(T \cap V)$:
\begin{itemize}
    \item[\textbf{1.}]if $0 \leq g(T) < \Upsilon$, then $\frac{\Upsilon}{2} \cdot c(T \cap V) = 0$;
    \item[\textbf{2.}]if $\Upsilon \leq g(T) < 2\Upsilon$, then $\frac{\Upsilon}{2} \cdot c(T \cap V) = \Upsilon$;
    \item[\textbf{3.}]if $g(T) \geq 2\Upsilon$, then $\frac{\Upsilon}{2} \cdot c(T \cap V) = 2 \Upsilon$.
\end{itemize}

Consider now \MUSRTJAR defined by
$g, X\triangleq \{x_1,x_2\}, Y\triangleq \{y_1,y_2\}$.
Note that $g(X)=g(Y) = 2 \Upsilon + f(\emptyset)$.
Since we are allowed to use \TAR and \TJ steps,
for any $S \subseteq [n]$,
we can construct a reconfiguration sequence $\calS$ whose value is 
$g(\calS) = \Upsilon + f(S)$:
an example of such a sequence is
$\langle \{x_1,x_2\}, \cdots $ adding elements of $S$ one by one $\cdots, \{x_1,x_2\} \cup S, \{x_1\} \cup S, \{y_1\} \cup S, \{y_1,y_2\} \cup S, \cdots $
removing elements of $S$ one by one $\cdots, \{y_1,y_2\} \rangle$.
Since we cannot transform from $\{x_1,x_2\}$ to $\{y_1,y_2\}$
without ever touching $T$ such that
$\frac{\Upsilon}{2} \cdot c(T \cap V) \leq \Upsilon$,
the optimal value for the \MUSRTJAR instance must be $\Upsilon + \OPT$.
Conversely, if a reconfiguration sequence $\calS$ has a value $g(\calS) > \Upsilon$,
we would be able to find a set $S^{(i)} \in \calS$ such that $ f(S^{(i)} \cap [n]) = g(\calS) - \Upsilon$.
In particular,
given a $(\frac{5}{6}+\epsilon')$-approximation algorithm for \MUSRTJAR,
we can find $(\frac{1}{2}+\epsilon)$-approximation for \USM by setting
$\epsilon' = \frac{4\epsilon}{9+6\epsilon} > 0$
because
$\frac{\Upsilon + (\frac{1}{2} +\epsilon)\OPT}{\Upsilon + \OPT}
\leq \frac{5}{6}+\epsilon'$.
Since no algorithm making fewer than $\rme^{\epsilon^2 n/8} $ oracle calls cannot find a $(\frac{1}{2}+\epsilon)$-approximation to \USM \cite[Theorem 4.5]{feige2011maximizing},
there is no $(\frac{5}{6}+\epsilon')$-approximation algorithm for \MUSRTJAR making fewer than
$\rme^{ (\frac{9 \epsilon'}{4-6\epsilon'})^2 \frac{n}{8} }$ oracle calls, which is more than
$\rme^{{\epsilon'}^2 n/2}$,
completing the proof.
\end{proof}

The last inapproximability result is presented below,
whose proof is similar to that of \cref{thm:usr-tjar-inapprox} and deferred to \cref{app:proofs}.
\begin{theorem}
\label{thm:usr-tar-inapprox}
For any $\epsilon > 0$,
there is no $(\frac{5}{6}+\epsilon)$-approximation algorithm for \MUSRTAR making at most $\rme^{\epsilon^2 n/2}$ oracle calls.
\end{theorem}

\section{Numerical Study}
\label{sec:exp}
We report numerical study on \MSR and \USRTJAR using real-world data.
We first applied \cref{alg:swap} for \MMSR to influence maximization reconfiguration.
We discover that \cref{alg:swap} quickly returns a reconfiguration sequence whose value is substantially better than the worst-case guarantee (\cref{subsec:infmax}).
We second applied \cref{alg:usr} for \USRTJAR to MAP inference reconfiguration on determinantal point processes.
We find that \cref{alg:usr}'s value is nine times smaller than the optimal value (\cref{subsec:map}).
We implemented an A* search algorithm for \MSR and \USR as a baseline
(see \cref{app:astar}
for details),
which was found to make more oracle calls than \cref{alg:swap,alg:usr}.
Experiments were conducted on a Linux server with Intel Xeon E5-2699 2.30GHz CPU and 792GB RAM.
All algorithms were implemented in Python~3.7.

\subsection{Influence Maximization Reconfiguration}
\label{subsec:infmax}
\subsubsection{\textbf{Problem Description}}
We formulate the reconfiguration of influence maximization as \MSR.
\emph{Influence maximization} \cite{kempe2003maximizing} requests to identify a fixed number of \emph{seed} vertices that maximize the spread of influence in a social network.
We adopt the \emph{independent cascade} model \cite{goldenberg2001talk} to
specify the process of network diffusion.
Given an \emph{influence graph} $G=(V,E,p)$,
where $p: E \to [0,1]$ is an edge probability function,
we consider the distribution over subgraphs $(V,E')$ obtained by maintaining each edge $e$ of $E$ with probability $p(e)$.
We say that a seed set $S \subseteq V$ \emph{activates}
a vertex $v \in V$ if $S$ can reach $v$ in $(V,E')$;
an objective function called the \emph{influence spread} $\Inf(S)$ is defined as 
the expected number of vertices that have been activated by $S$.
Since $\Inf(\cdot)$ is monotone and submodular \cite{kempe2003maximizing},
the reconfiguration version of influence maximization corresponds to \MSR, whose motivation was described in \cref{sec:intro}.

\subsubsection{\textbf{Setup}}
We prepare an influence graph $G=(V,E,p)$ and two input sets $X$ and $Y$.
We used two publicly-available social network data,
\karate network\footnote{\url{http://konect.cc/networks/ucidata-zachary/}} with $34$ vertices and $78$ bidirectional edges, and
\physicians network\footnote{\url{http://konect.cc/networks/moreno_innovation/}} with $117$ vertices and $542$ directed edges, from Koblenz Network Collection \cite{kunegis2013konect}.
We set the probability $p(u,v)$ of edge $(u,v) \in E$ to the inverse of the in-degree of $v$, which was adopted in  \cite{tang2014influence,arora2017debunking,ohsaka2020solution}.
Since exact computation of $\Inf(\cdot)$ is \shP-hard \cite{chen2010scalable},
we used the approximation scheme in \cite[\S 5.2]{ohsaka2020solution} to
construct a monotone submodular function $f:2^V \to \bbR_+$
from $10^5$ reverse reachable sets \cite{borgs2014maximizing,tang2014influence},
which provides an \emph{unbiased estimate} for the influence spread.
We constructed $X$ and $Y$ 
so that they are disjoint and moderately influential.
To that end, we ran the greedy algorithm \emph{interchangeably}:
Beginning with $X_0 \triangleq \emptyset$ and $Y_0 \triangleq \emptyset$,
we compute $X_i$ and $Y_i$ for $i \geq 1$ as follows:
\begin{small}
\begin{align}
\label{eq:exp-subsets}
    X_i & \triangleq X_{i-1} \cup \Bigl\{ \argmax_{e \in [n] \setminus (X_{i-1} \cup Y_{i-1})} f(X_{i-1} \cup \{e\}) \Bigr\}, \\
    Y_i & \triangleq Y_{i-1} \cup \Bigl\{ \argmax_{e \in [n] \setminus (X_i \cup Y_{i-1})} f(Y_{i-1} \cup \{e\}) \Bigr\}.
\end{align}
\end{small}
On \karate, we define
$X \triangleq X_{8}$ and $Y \triangleq Y_{8}$,
where $f(X) = 23.2$ and $f(Y) = 23.6$,
which are drawn in \cref{fig:karate-vis}.
On \physicians, we define
$X \triangleq X_{16}$ and $Y \triangleq Y_{16}$,
where $f(X) = 93.2$ and $f(Y) = 93.3$.

\input{fig-msreco}

\subsubsection{\textbf{Results}}
We ran \cref{alg:swap} and the A* algorithm with
$\theta = 0.9v, 0.95v, v$,
where $v \triangleq \min\{f(X), f(Y)\}$,
on \karate and \physicians.
The obtained sequences were found to be all the \emph{shortest}.
\cref{fig:msreco} displays
the influence spread of sets in each reconfiguration sequence.
On \karate, the A* algorithm with $\theta=v$ and \cref{alg:swap} found an optimal reconfiguration sequence of value $v$.
We can observe that
the intermediate sets for \cref{alg:swap} were more influential than those for the A* algorithm.
\cref{fig:karate-vis} draws \karate network, where each vertex is colored according to its probability of being activated by
$X$, the fourth subset $S^{(4)}$ returned by \cref{alg:swap}, or $Y$.
We can see that
many vertices are more likely to be activated by $S^{(4)}$ than by $X$ or $Y$, making it easy to transform from $X$ to $Y$.
(See \cref{app:table}
for the entire reconfiguration sequence returned by \cref{alg:swap}.)
On \physicians, \cref{alg:swap} found
a reconfiguration sequence of value $ 84.9 \approx 0.91 v$,
which is still drastically better than $\frac{1}{2} v$
envisioned from \cref{thm:msr-approx}, though the A* algorithm's sequence has value $v$.
We finally report the number of oracle calls for an influence function:
On \karate,
\cref{alg:swap} made $72$ calls and the A* algorithm made $1{,}428$ calls;
on \physicians,
\cref{alg:swap} made $272$ calls and the A* algorithm made $53{,}987$ calls.
(We stress that we do not report actual running time as it heavily depends on implementations of an unbiased estimator \cite{ohsaka2020solution} and
scalability against large instances is beyond the scope of this paper.)
In summary,
\cref{alg:swap} produced
a reconfiguration sequence of reasonable quality
by making fewer oracle calls.

\subsection{MAP Inference Reconfiguration}
\label{subsec:map}
\subsubsection{\textbf{Problem Description}}
We formulate the reconfiguration of maximum a posteriori (MAP) inference on determinantal point process as \USRTJAR.
\emph{Determinantal point processes} (DPPs) \cite{macchi1975coincidence,borodin2005eynard} are
a probabilistic model on the power set $2^{[n]}$, which captures negative correlations among objects. 
Given a Gram matrix $\mat{A} \in \bbR^{n \times n}$,
a DPP defines the probability mass of each subset $S \subseteq [n]$ to be proportional to $\det(\mat{A}_S)$.
Seeking a subset with the maximum determinant
(i.e., $\max_{S \subseteq [n]} \det(\mat{A}_S)$),
which is equivalent to \emph{MAP inference} \cite{gillenwater2012near},
finds applications in recommendation and summarization \cite{wilhelm2018practical,yao2016tweet,kulesza2012determinantal}.
Since $\log \det(\mat{A}_S)$ as a set function in $S$ is submodular,
the reconfiguration counterpart of MAP inference is \USR, whose motivation was explained in \cref{sec:intro}.

\input{fig-usreco}

\subsubsection{\textbf{Setup}}
We prepare a Gram matrix $\mat{A}$ and a pair of input sets $X$ and $Y$.
We used \texttt{MovieLens 1M}\footnote{\url{https://grouplens.org/datasets/movielens/1m/}} \cite{harper2015movielens},
which consists of
$1$ million ratings on $3{,}900$ movies from $6{,}040$ users of
an online movie recommendation website MovieLens.\footnote{\url{http://movielens.org/}}
We first selected $n = 207$ movies with at least $1{,}000$ ratings and
$m = 839$ users who rated at least $100$ movies,
resulting in an $n \times m$ movie-user rating matrix.
We then ran Nonnegative Matrix Factorization \cite{boutsidis2008svd} with dimension $64$ to extract
a feature vector $\vec{\phi}_i \in \bbR^{64}$
with $\|\vec{\phi}_i\|_2 = 1$
for each movie $i \in [n]$.
The Gram matrix $\mat{A} \in \bbR^{n \times n}$ is constructed as
$A_{i,j} = \langle 2^{r_i - 4} \vec{\phi}_i, 2^{r_j-4} \vec{\phi}_j \rangle$ for all $i,j \in [n]$,
where $r_i$ is an average rating of movie $i$ between $[1, 5]$.
Since
$\det(\mat{A}_S)$ is equal to
$(\prod_{i \in S} 2^{r_i - 4})^2$ times the square volume of the parallelepiped spanned by $\{\vec{\phi}_i\}_{i \in S}$ \cite{kulesza2012determinantal},
movies in a subset of large determinant are expected to be
\emph{highly-rated} and of \emph{diverse} genres.
An input submodular function $f$ is defined as $f(S) \triangleq \log \det(\mat{A}_S)$ for $S \subseteq [n]$.
We created $X$ and $Y$ in a similar manner to the first experiment:
We computed $X_i$ and $Y_i$ for $i \geq 1$ according to \cref{eq:exp-subsets}
until no further selection is possible and
extracted those with the largest determinant, resulting is that
$X \triangleq X_{24}$ and
$Y \triangleq Y_{22}$,
where $f(X) = 8.52$ and $f(Y) = 7.79$.

\subsubsection{\textbf{Results}}
We ran \cref{alg:usr} and the A* algorithm with
$\theta = 0.9 v, 0.95 v, v$, where
$v \triangleq \min\{f(X), f(Y)\}$.
The A* algorithm produced reconfiguration sequences of length $24$ while \cref{alg:usr} produced a reconfiguration sequence of length $45$.
\cref{fig:usreco} plots the log-determinant of sets in each reconfiguration sequence.
The A* algorithm with $\theta = v$
was able to find an optimal reconfiguration sequence $\calS^*$.
$22$ of the $24$ steps in $\calS^*$ were found to be \TJ steps,
which is quite different from the behavior of \cref{alg:usr}.
One possible reason is that
log-determinant functions exhibit \emph{monotonicity}
when every eigenvalue of $\mat{A}$ is greater than $1$ \cite{sharma2015greedy};
in fact, the principal submatrix of $\mat{A}$ induced by $X$ and $Y$
has the minimum eigenvalue of $0.70$ and $0.76$, respectively,
while the minimum eigenvalue of $\mat{A}$ was approximately $0$.
Hence, there is a sequence of \TJ steps that preserves the log-determinant large.
As opposed to the success of \cref{alg:swap} for \MSR, \cref{alg:usr}'s value was $0.823 \approx 0.11 v$, which is nine times smaller than the optimal value $v$.
This result is easily expected from the mechanism of \cref{alg:usr}, which includes \emph{singletons} (i.e., $\{x_1\}$ and $\{y_1\}$)
into the output sequence.
The number of oracle calls for a log-determinant function was $553$ for \cref{alg:usr} and $105{,}412$ for the A* algorithm.\footnote{
Again, we do not report actual running time, which is severely affected by implementation of determinant computation \cite{chen2018fast}.}
We conclude that \cref{alg:usr} consumes fewer oracle calls but further development on approximation algorithms for \MUSRTJAR is required.

\section{Conclusion and Open Questions}
We established an initial study on submodular reconfiguration problems, including
intractability, (in)approximability, and numerical results.
We conclude this paper with two open questions.

\begin{itemize}
\item Can we devise an approximation algorithm for \MUSRTAR?
\item Can the approximation factors in \cref{sec:approx} be made tight? We conjecture
an $\bigO(1)$-factor approximability for \MUSRTJAR.
\end{itemize}

%% file: summary.tex
\begin{table*}[tbp]
    \centering
    \caption{Complexity-theoretic results of submodular reconfiguration problems (see \cref{subsec:defining} for definitions).}
    \label{tab:summary}
    \setlength{\tabcolsep}{2pt}
    \renewcommand{\arraystretch}{0.7}
    \begin{tabular}{c|ccc|ccc}
    \toprule
    \textbf{name} & \textbf{source} & \textbf{transformation} & \textbf{\cref{sec:problem} formulation} & \textbf{\cref{sec:hardness} exact solution} & \textbf{\cref{sec:approx} approximability} & \textbf{\cref{sec:inapprox} inapproximability} \\
    \midrule
    \multirow{2}{*}{\MSR} &
    \multirow{2}{*}{$\displaystyle\max_{S \in {[n] \choose k}} f(S)$} &
    \multirow{2}{*}{jump} &
    \cref{prb:msr} &
    \PSPACE-complete  &
    $\max\{\frac{1}{2}, (1-\kappa)^2\}$-factor &
    $(1-\frac{1+\epsilon}{n^2})$-factor \footnotesize{$\Rightarrow$ no FPTAS} \\
    &
    &
    &
    \cref{prb:mmmsr} &
    (\cref{thm:msr-pspace}) &
    (\cref{thm:msr-approx}) &
    (\cref{thm:msr-inapprox}) \\
    \midrule

    \multirow{2}{*}{\USRTAR} &
    \multirow{2}{*}{$\displaystyle\max_{S \subseteq [n]} f(S)$} &
    \multirow{2}{*}{add/remove} &
    \cref{prb:usr-tar} &
    \PSPACE-complete &
    \textit{open} &
    $(\frac{5}{6}+\epsilon)$-factor \\
    &
    &
    &
    \cref{prb:mmusr-tar} &
    (\cref{thm:usr-tar-pspace}) &
    (see also \cref{subsec:usr-tar-approx}) &
    (\cref{thm:usr-tar-inapprox}) \\
    \midrule

    \multirow{2}{*}{\USRTJAR} &
    \multirow{2}{*}{$\displaystyle\max_{S \subseteq [n]} f(S)$} &
    \multirow{2}{*}{jump/add/remove} &
    \cref{prb:usr-tjar} &
    \PSPACE-complete &
    $\frac{1}{n}$-factor &
    $(\frac{5}{6}+\epsilon)$-factor \\
    &
    &
    &
    \cref{prb:mmusr-tjar} &
    (\cref{thm:usr-tjar-pspace}) &
    (\cref{thm:usr-tjar-approx}) &
    (\cref{thm:usr-tjar-inapprox}) \\
    \bottomrule
    \end{tabular}
    
    \small{$n$ denotes the size of the ground set; $\kappa$ denotes the total curvature of an input submodular function;
    $\epsilon$ is an arbitrarily small positive number.}
\end{table*}

%% file: fig-karate-vis.tex
\begin{figure}
    \centering
    \includegraphics[width=0.55\hsize]{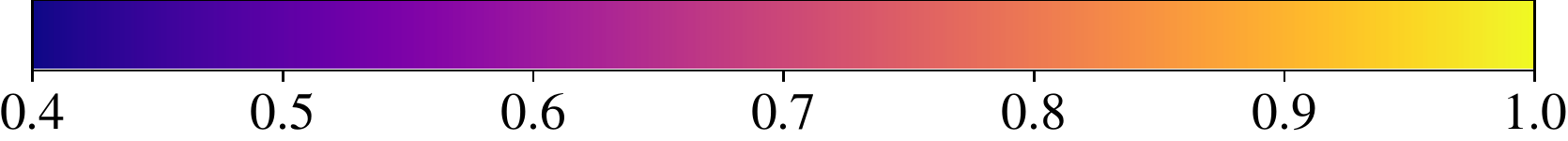}%
    \vspace{-1\baselineskip}
    \subfloat[\scriptsize{$f(X) = 23.2$ \\ $X = \{2, 4, 7, 9, 11, 30, 32, 34\}$}
    \label{fig:karate-S}]{%
    \includegraphics[width=0.333\hsize]{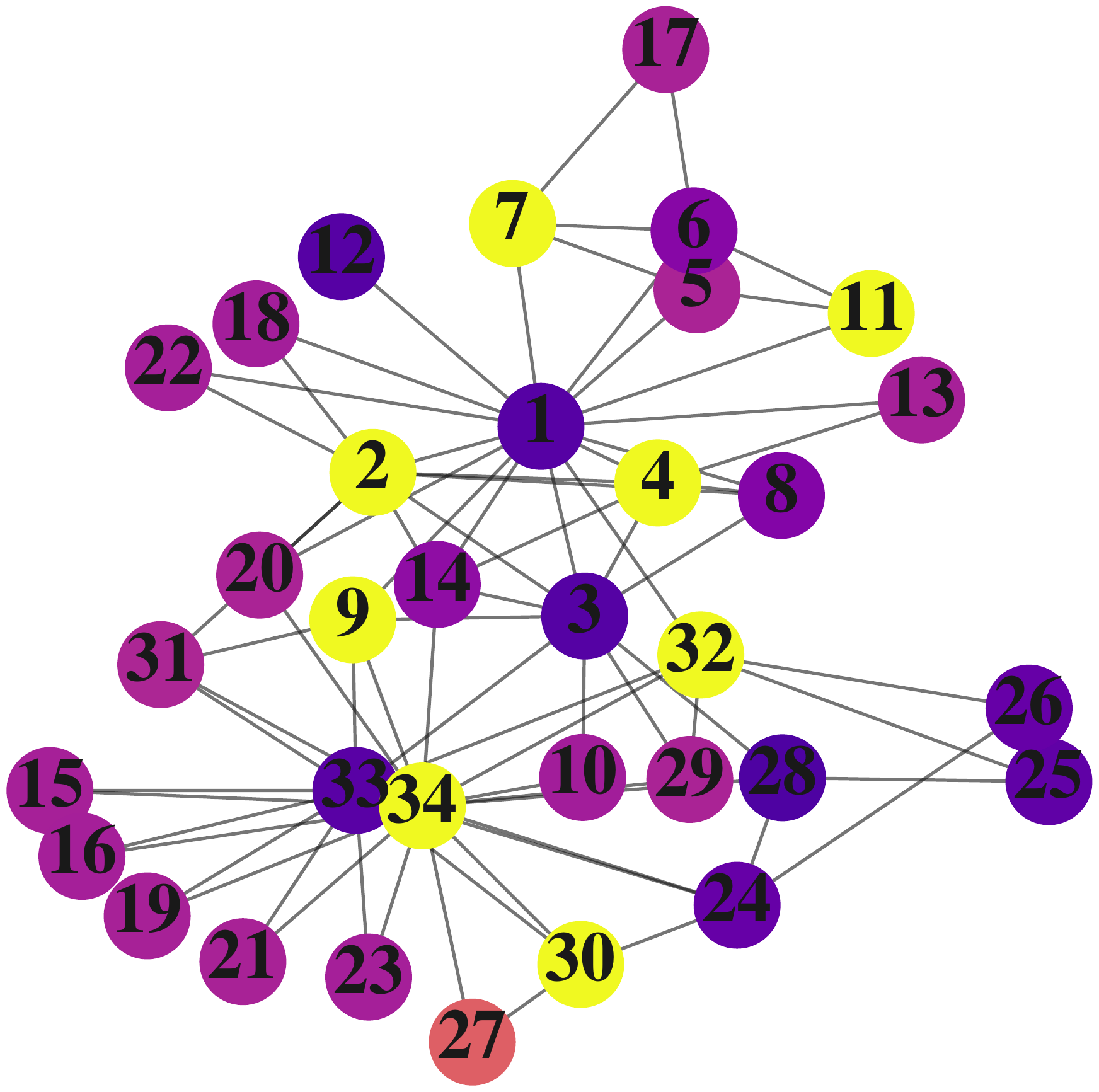}%
    }%
    \subfloat[\scriptsize{$f(S^{(4)}) = 25.5$ \\ $S^{(4)} = \{1, 2, 3, 6, 7, 32, 33, 34\}$}
    \label{fig:karate-S4}]{%
    \includegraphics[width=0.333\hsize]{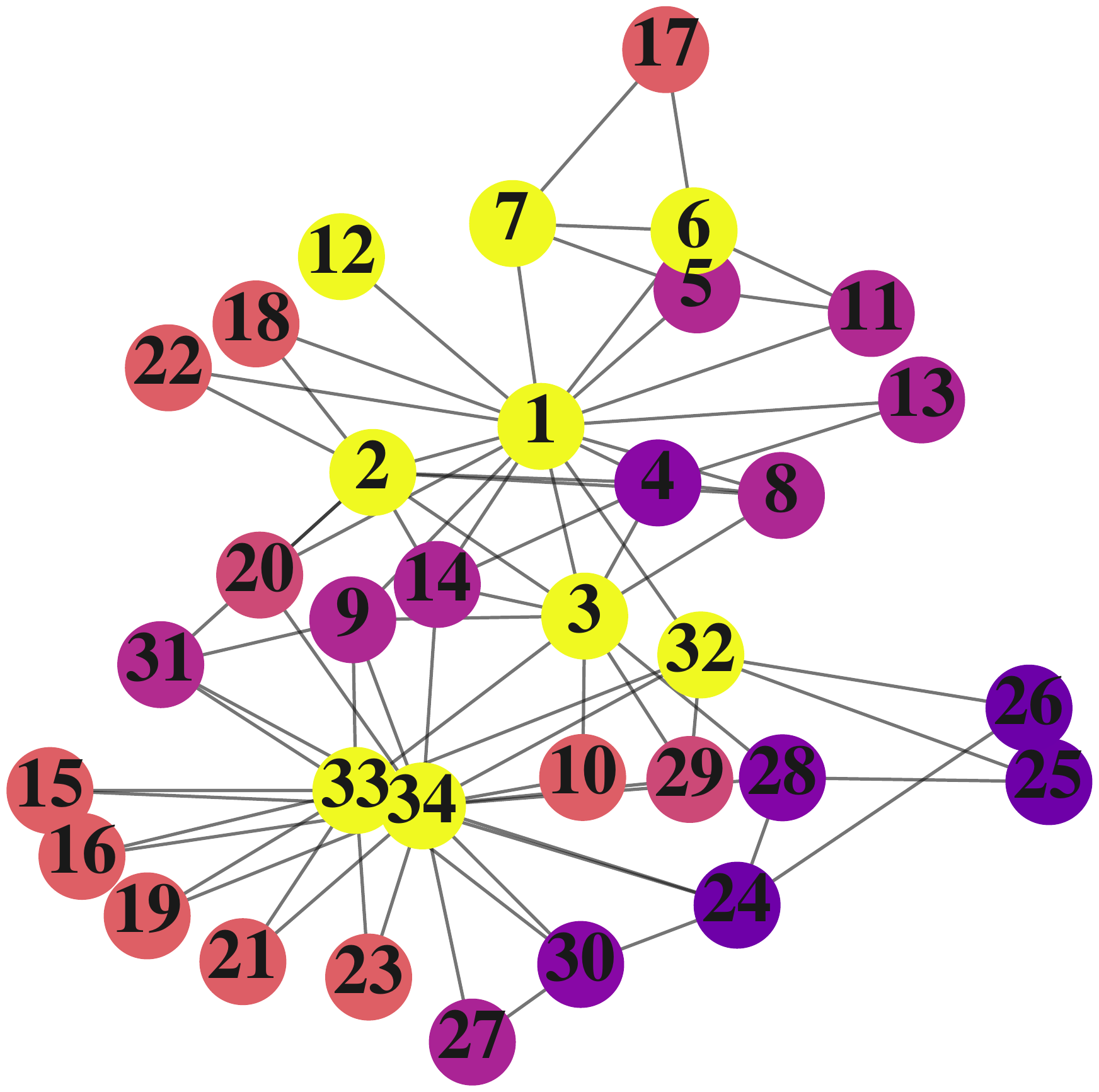}%
    }%
    \subfloat[\scriptsize{$f(Y) = 23.6$ \\ $Y = \{1, 3, 6, 14, 24, 25, 27, 33\}$}
    \label{fig:karate-T}]{%
    \includegraphics[width=0.333\hsize]{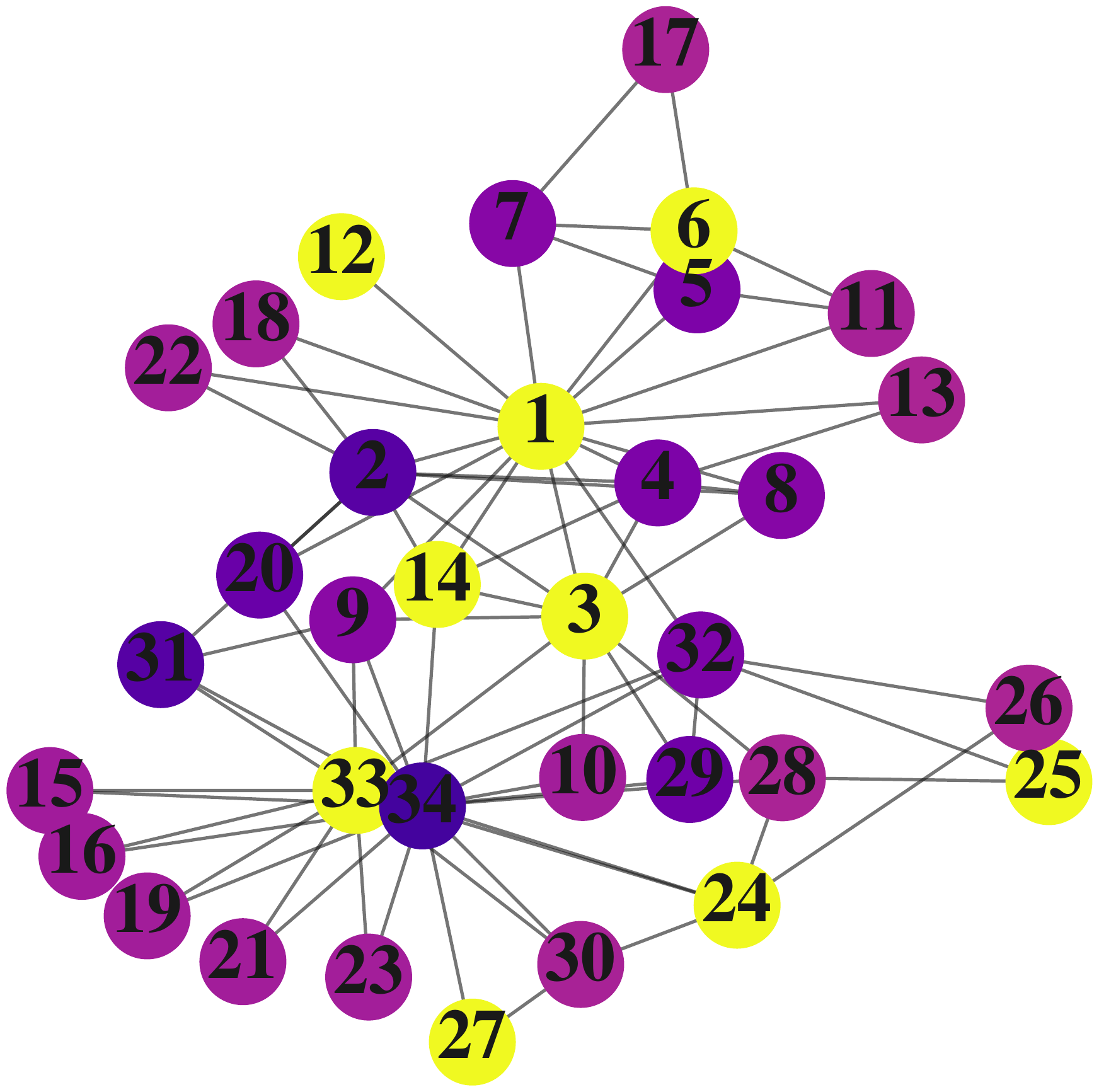}%
    }%
    \caption{Example of influence maximization reconfiguration on \karate network (see \cref{subsec:infmax}).
    Each vertex is colored according to the activation probability.
    Given seed sets $X$ (\ref{fig:karate-S}) and $Y$ (\ref{fig:karate-T}),
    we wish to find a sequence of influential seed sets connecting them.
    The 4th seed set $S^{(4)}$ (\ref{fig:karate-S4}) found by our algorithm (see \cref{subsec:msr-approx})
    is more influential than $X$ and $Y$.
    }
    \label{fig:karate-vis}
\end{figure}

%% file: alg-greedy.tex
\begin{algorithm}[tbp]
\caption{Greedy algorithm.}
\label{alg:greedy}
\small
\begin{algorithmic}[1]
    \Require
    function $f: 2^{[n]} \to \bbR_+$;
    set $N \subseteq [n]$;
    solution size $k \leq |N|$.
    \State \textbf{for each } $i = 1 \To k$ \textbf{do } $\displaystyle e_i \leftarrow \argmax_{e \in N \setminus \{ e_1, \ldots, e_{i-1} \}} f(\{e_1, \ldots, e_{i-1}, e\}) $.
    \State \Return sequence $\langle e_1, \ldots, e_k \rangle$.
\end{algorithmic}
\end{algorithm}

%% file: alg-swap.tex
\begin{algorithm}[tbp]
\caption{$\max\{\frac{1}{2}, (1-\kappa)^2\}$-approximation algorithm for \MMSR.}
\label{alg:swap}
\small
\begin{algorithmic}[1]
    \Require
    monotone submodular func.~$f: 2^{[n]} \to \bbR_+$;
    two sets $X, Y \in {[n] \choose k}$.
    \State $R \leftarrow X \cap Y, \quad X' \leftarrow X \setminus R, \quad Y' \leftarrow Y \setminus R, \quad k' \leftarrow |X'| = |Y'|$.
    \State invoke \cref{alg:greedy} on $f_R,X',k'$ to get greedy sequence
    $\langle x_1, \ldots, x_{k'} \rangle$.
    \State invoke \cref{alg:greedy} on $f_R,Y',k'$ to get greedy sequence $\langle y_1, \ldots, y_{k'} \rangle$.
    \State \textbf{for each} $i = 0 \To k'$ \textbf{do} $S^{(i)} \leftarrow \{x_1, \ldots, x_{k'-i}\} \uplus \{ y_1, \ldots, y_i \} \uplus R$.
    \State \Return sequence $\calS = \langle S^{(0)}, \ldots, S^{(k')} \rangle$.
\end{algorithmic}
\end{algorithm}

%% file: alg-usr.tex
\begin{algorithm}[tbp]
\caption{$\frac{1}{n}$-approximation algorithm for \MUSRTJAR.}
\label{alg:usr}
\small
\begin{algorithmic}[1]
    \Require
    submodular function $f: 2^{[n]} \to \bbR_+$;
    two subsets $X, Y$ of $[n]$.
    \State invoke \cref{alg:greedy} on $f,X,|X|$ to
    get greedy sequence $\langle x_1, \ldots, x_{|X|} \rangle$.
    \State invoke \cref{alg:greedy} on $f,Y,|Y|$ to
    get greedy sequence $\langle y_1, \ldots, y_{|Y|} \rangle$.
    \State declare empty sequence $\calS = \langle \rangle$.
    \State \textbf{for each} $i = |X| \To 1$ \textbf{do} append $\{x_1, \ldots, x_i\}$ at the end of $\calS$.
    \State \textbf{for each} $i = 1 \To |Y|$ \textbf{do} append $\{y_1, \ldots, y_i\}$ at the end of $\calS$.
    \State \Return sequence $\calS$.
\end{algorithmic}
\end{algorithm}

%% file: fig-msreco.tex
\begin{figure}
    \centering
    \subfloat[\karate network]{%
    \includegraphics[width=0.5\hsize]{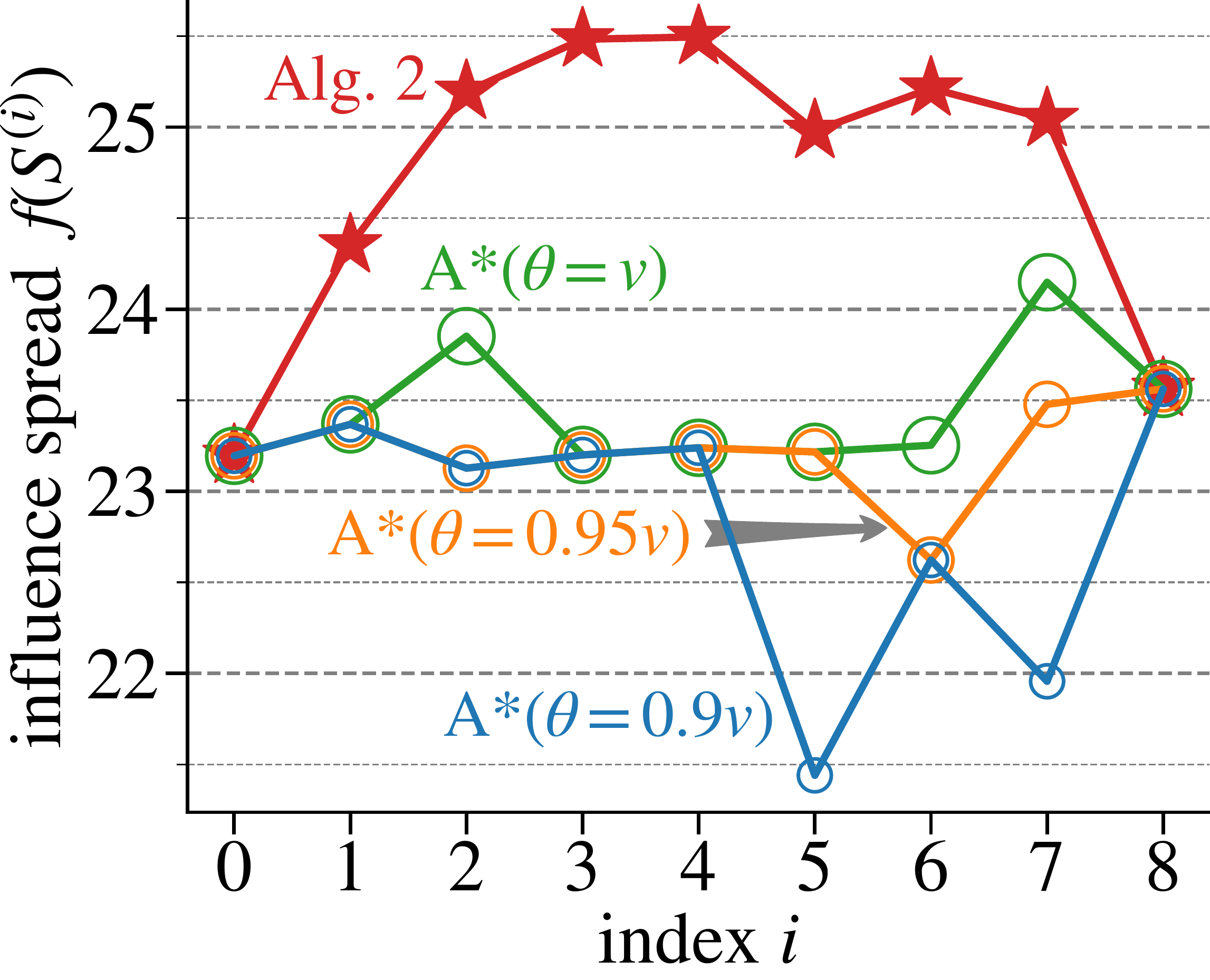}%
    }%
    \subfloat[\physicians network]{%
    \includegraphics[width=0.5\hsize]{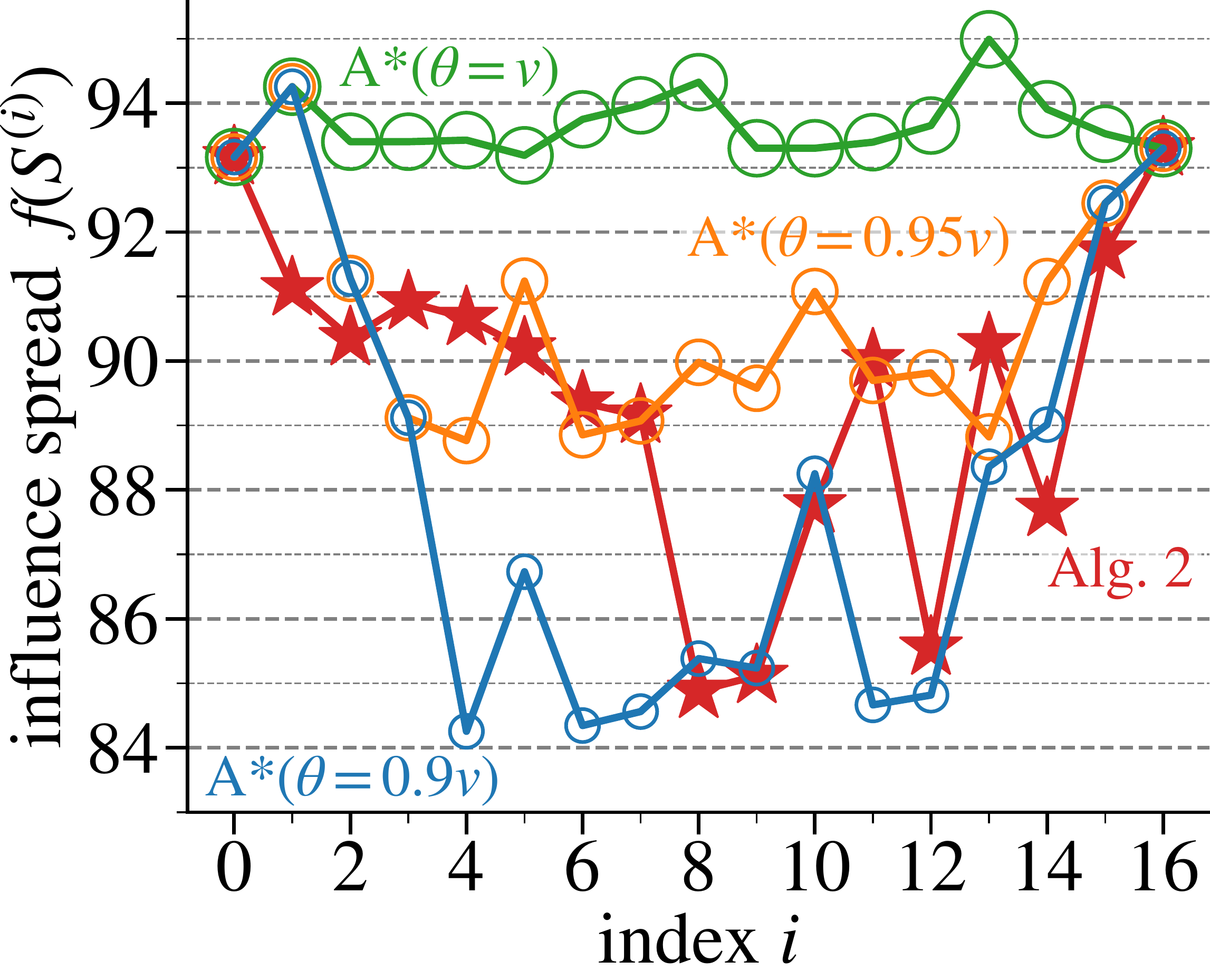}%
    }%
    \caption{Influence spread of each set in the reconfiguration sequences returned by A* algorithm ($\theta=0.9v,0.95v,v$) and \cref{alg:swap}, where  $v = \min\{f(X), f(Y)\}$.}
    \label{fig:msreco}
\end{figure}

%% file: fig-usreco.tex
\begin{figure}
    \centering
    \includegraphics[width=1\hsize]{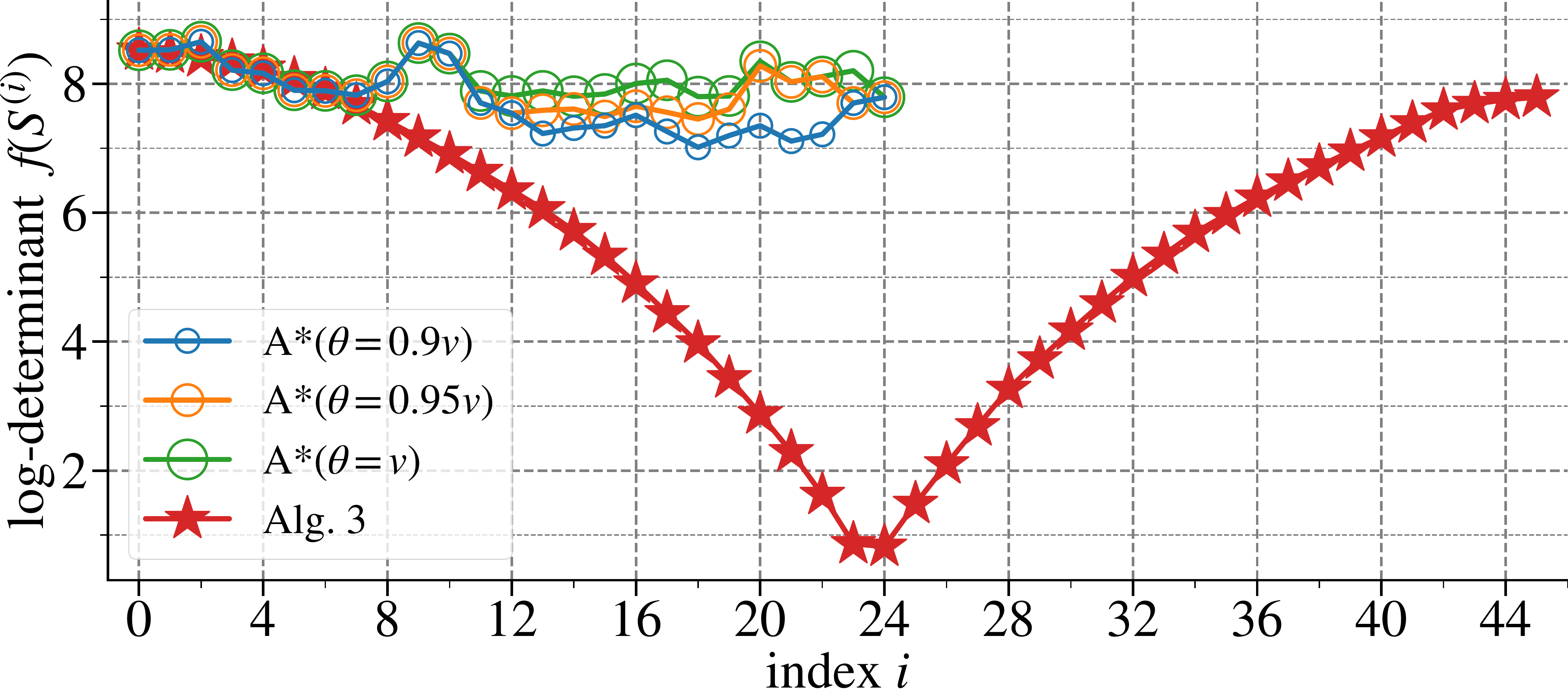}
    \caption{Log-determinant of each set in the reconfiguration sequences returned by A* algorithm ($\theta=0.9v,0.95v,v$) and \cref{alg:usr} on \texttt{MovieLens 1M}, where $v = \min\{f(X),f(Y)\}$.}
    \label{fig:usreco}
\end{figure}

%% file: app.tex
\section{Missing Proofs}
\label{app:proofs}
\begin{proof}[Proof of \cref{obs:pspace}]
It is known \cite{heuvel13complexity} that 
a reachability problem defined in the reconfiguration framework
is in \NPSPACE if the following assumptions hold:
\begin{itemize}
    \item[\textbf{1.}]given a possible solution, we can determine whether it is feasible in polynomial time;
    \item[\textbf{2.}]given two feasible solutions, we can decide if there is a reconfiguration step from one to the other in polynomial time.
\end{itemize}
It is easy to see that \cref{prb:msr,,prb:usr-tar,,prb:usr-tjar} meet these assumptions.
By Savitch's theorem \cite{savitch1970relationships},
we have that \PSPACE~$=$~\NPSPACE, which completes the proof.
\end{proof}

To prove \PSPACE-hardness of \MVCR (\cref{lem:rvcr-pspace}),
we use a reduction from \SATR \cite{gopalan2009connectivity}.
Given a \emph{3-conjunctive normal form (3-CNF) formula} $\phi$,
of which each clause contains at most three literals\footnote{
Without loss of generality, we can assume that no clause contains both positive and negative literals of the same variable.}
(e.g., $\phi = (x_1 \vee \overline{x_2} \vee x_3) \wedge (x_2 \vee \overline{x_3} \vee x_4) \wedge (\overline{x_1} \vee x_3 \vee \overline{x_4})$),
\prb{3-SAT} asks to decide if there exists a truth assignment $\vec{\sigma}$ for the variables of $\phi$ that satisfies all clauses of $\phi$
(e.g., $\vec{\sigma}(x_1) = \vec{\sigma}(x_2) = \True$ and $\vec{\sigma}(x_3) = \vec{\sigma}(x_4) =\False$).
\SATR is defined as:

\begin{problem}[\SATR \cite{gopalan2009connectivity}]
\label{prb:satr}
Given a 3-CNF formula $\phi$ and
two satisfying truth assignments $\vec{\sigma}^x$ and $\vec{\sigma}^y$ of $\phi$,
determine whether there exists a sequence of satisfying truth assignments of $\phi$ from $\vec{\sigma}^x$ to $\vec{\sigma}^y$,
$\langle \vec{\sigma}^{(0)} = \vec{\sigma}^x, \vec{\sigma}^{(1)}, \ldots, \vec{\sigma}^{(\ell)}= \vec{\sigma}^y \rangle$,
such that
each truth assignment is obtained from the previous one by a single variable flip; i.e.,
they differ in exactly one variable.
\end{problem}
\prb{3-SAT} is widely known to be \NP-complete \cite{cook1971complexity,levin1973universal} while
\SATR is \PSPACE-complete \cite{gopalan2009connectivity}.

\begin{proof}[Proof of \cref{lem:rvcr-pspace}]
The proof mostly follows \cite{ito2011complexity}.
We show a polynomial-time reduction from \SATR.
Suppose we are given
a 3-CNF formula $\phi$ with $n$ variables $x_1, \ldots, x_n$ and $m$ clauses $c_1, \ldots, c_m$ and
two satisfying truth assignments $\vec{\sigma}^x$ and $\vec{\sigma}^y$ of $\phi$.
Starting with an empty graph,
we construct a graph $G_\phi$ in polynomial time according to \cite[Proof of Theorem 2]{ito2011complexity}:
\begin{itemize}
\item \textbf{Step 1.}~for each variable $x_i$ in $\phi$, we add an edge to $G_\phi$,
the endpoints of which are labeled $x_i$ and $\overline{x}_i$;
\item \textbf{Step 2.}~for each clause $c_j$ in $\phi$, we add a clique of size $|c_j|$ to $G_\phi$,
each vertex in which corresponds to a literal in $c_j$;
\item \textbf{Step 3.}~we connect between two vertices in different components by an edge
if they correspond to opposite literals of the same variable,
e.g., $x_i$ and $\overline{x}_i$.
\end{itemize}
It is proven \cite{ito2011complexity} that
$G_\phi$ has a maximum independent set\footnote{
An \emph{independent set} is a set of vertices in which no pair of two vertices are adjacent.
} $I$ of size $m+n$
\emph{if and only if}
$\phi$ is satisfiable; here,
$n$ vertices in $I$ are chosen from the endpoints of the $n$ edges corresponding to the variables of $\phi$, and
$m$ vertices in $I$ are chosen from the $m$ cliques corresponding to the clauses of $\phi$.
Using such $I$,
we can uniquely construct a satisfying truth assignment $\vec{\sigma}_I$ of $\phi$, which
assigns $\True$ (resp.~$\False$) to $x_i$ if $I$ includes the endpoint labeled $x_i$ (resp.~$\overline{x_i}$).
On the other hand, for a fixed truth satisfying assignment $\vec{\sigma}$,
there may be exponentially many maximum independent sets $I$ such that $\vec{\sigma}_I = \vec{\sigma}$.
Observe now the following facts for $G_\phi$
for which $\phi$ is satisfiable:
\begin{itemize}
    \item For any two satisfying truth assignments $\vec{\sigma}_1$ and $\vec{\sigma}_2$ that differ in exactly one variable,
    there exist two maximum independent sets $I_1$ and $I_2$
    such that
    $I_1$ is obtained from $I_2$ by a single \TJ step,
    $\vec{\sigma}_{I_1} = \vec{\sigma}_1$ and $\vec{\sigma}_{I_2} = \vec{\sigma}_2$.
    \item
    For any two maximum independent sets $I_1$ and $I_2$ corresponding to the same satisfying truth assignments (i.e., $\vec{\sigma}_{I_1} = \vec{\sigma}_{I_2}$),
    there exists a sequence of maximum independent sets corresponding to the same satisfying truth assignment from $I_1$ to $I_2$
    under \TJ.
\end{itemize}

We now translate the above discussion into the language of \emph{vertex cover}.
Because a vertex set $C \subseteq V(G_\phi)$ is a minimum vertex cover of $G_\phi$ \emph{if and only if} $V(G_\phi) \setminus C$ is a maximum independent set of $G_\phi$,
$G_\phi$ has a minimum vertex cover $C$ of size $|V(G_\phi)|-m-n$
\emph{if and only if} $\phi$ is satisfiable;
we can uniquely construct a satisfying truth assignment $\vec{\sigma}_{V(G_\phi) \setminus C}$.
Consequently, for any two minimum vertex covers $C^x$ and $C^y$ of size $|V(G_\phi)|-m-n$ such that
$\vec{\sigma}_{V(G_\phi) \setminus C^x} = \vec{\sigma}^x$ and $\vec{\sigma}_{V(G_\phi) \setminus C^y} = \vec{\sigma}^y$,
there exists a sequence of satisfying truth assignments
from $\vec{\sigma}^x$ to $\vec{\sigma}^y$ that
meets the specification for \SATR
\emph{if and only if}
there exists a sequence of minimum vertex covers 
from $C^x$ to $C^y$ under \TJ.
Such vertex covers $C^x$ and $C^y$ can be found in polynomial time, which completes the reduction from \SATR to \MSR.
\end{proof}

\begin{proof}[Proof of \cref{thm:usr-tjar-pspace}]
We demonstrate a polynomial-time reduction from
\cref{prb:rvcr}.
Given a graph $G=(V,E)$ and two minimum vertex covers $C^x$ and $C^y$ of size $k$,
we define a submodular function $f: 2^{V} \to \bbR_+$ such that
$f(S)$ is ``the number of edges in $E$ that are incident to $S$ \emph{minus}
$\frac{1}{2}(n - |S|)$,'' where $n \triangleq |V|$.
Consider \USRTJAR defined by
$f$, $C^x$, $C^y$, and
a threshold $\theta \triangleq |E|-\frac{k}{2}+\frac{n}{2}$.
It turns out that any reconfiguration sequence from $C^x$ to $C^y$ does not include a vertex set of size either $k-1$ or $k+1$---that is, we can only apply  \TJ steps---in the following case analysis on $f(S)$:
\begin{itemize}
    \item[\textbf{1.}]if $|S|=k-1$ ($S$ cannot be a vertex cover): $f(S) \leq |E|-\frac{k}{2}+\frac{n}{2}-\frac{1}{2}$;
    \item[\textbf{2.}]if $|S|=k$ and $S$ is a vertex cover: $f(S) = |E|-\frac{k}{2}+\frac{n}{2}$;
    \item[\textbf{3.}]if $|S|=k$ and $S$ is not a vertex cover: $f(S) \leq |E|-\frac{k}{2}+\frac{n}{2}-1$;
    \item[\textbf{4.}]if $|S|=k+1$ ($S$ may be a vertex cover): $f(S) \leq |E|-\frac{k}{2}+\frac{n}{2}-\frac{1}{2}$.
\end{itemize}
Therefore, a reconfiguration sequence on the \MVCR instance is
a reconfiguration sequence on the \USRTJAR instance, and vice versa,
which completes the reduction;
the \PSPACE-hardness follows from \cref{lem:rvcr-pspace}.
\end{proof}

\begin{proof}[Proof for $(1-\kappa)^2$-approximation in \cref{thm:msr-approx}]
We reuse the notations $R, X', Y', k', x_i, y_i, X_i, Y_i$
from the proof of \cref{thm:msr-approx} in the main body.
We denote by $\kappa$ the total curvature of $f$; note that the residual $f_R$ has a total curvature not more than $\kappa$.

Showing that $f(\calS) \geq (1-\kappa)^2\min\{f(X),f(Y)\}$ is sufficient.
For each $i \in [k']$, we denote
$\Delta x_i \triangleq f_R(X_i) - f_R(X_{i-1})$ and
$\Delta y_i \triangleq f_R(Y_i) - f_R(Y_{i-1})$.
Note that $\Delta x_i$ and $\Delta y_i$ are monotonically nonincreasing in $i$ due to \cref{eq:greedy-ineq}; i.e.,
$\Delta x_1 \geq \Delta x_2 \geq \cdots \geq \Delta x_{k'}$ and
$\Delta y_1 \geq \Delta y_2 \geq \cdots \geq \Delta y_{k'}$.
We define a set function $\overline{f_R}: 2^{[n] \setminus R} \to \bbR_+$ such that, for each $S \subseteq [n]\setminus R$,
$
    \overline{f_R}(S) \triangleq \sum_{e \in S} f_R(\{e\}).
$
Note that $\overline{f_R}$ is a monotone modular function, and that
$\overline{f_R}(S)$ gives an upper bound of $f_R(S)$.
Moreover, $\overline{f_R}$ gives a $(1-\kappa)$-factor approximation to $f_R$ (e.g., \cite[Lemma 2.1]{iyer2013curvature}); i.e.,
\begin{align}
\label{eq:modular}
    (1-\kappa)\overline{f_R}(S) \leq f_R(S) \leq \overline{f_R}(S), \text{ for all } S \subseteq [n] \setminus R.
\end{align}
We will bound $\overline{f_R}(X_{k'-i} \uplus Y_i)$ from below for each $i \in [k'-1]$ in a case analysis.
We have two cases to consider:
\begin{itemize}
    \item[\textbf{1.}]$\Delta y_i \geq \Delta x_{k'-i+1}$.
    We then have that
    $\Delta y_1 + \cdots + \Delta y_i \geq \Delta x_{k'-i+1} + \cdots + \Delta x_k$.
    By adding $\sum_{1 \leq j \leq k'-i} \Delta x_j$ to both sides, we obtain that
    \begin{align*}
        \sum_{1 \leq j \leq i} \Delta y_j + \sum_{1 \leq j \leq k'-i} \Delta x_j
        & \geq \sum_{k'-i+1 \leq j \leq k'} \Delta x_j + \sum_{1 \leq j \leq k'-i} \Delta x_j \\
        \Rightarrow f_R(Y_i) + f_R(X_{k'-i}) & \geq f_R(X_{k'}) + f_R(\emptyset).
    \end{align*}
    Simple calculation using \cref{eq:modular} yields that
    $
        \overline{f_R}(Y_i \uplus X_{k'-i}) = \overline{f_R}(Y_i) + \overline{f_R}(X_{k'-i}) \geq (1-\kappa)\overline{f_R}(X'),
    $
    where we note that $\overline{f_R}(\emptyset) = 0$.
    \item[\textbf{2.}]$\Delta y_i \leq \Delta x_{k'-i+1}$.
    We then have that
    $\Delta x_1 + \cdots + \Delta x_{k'-i} \geq \Delta y_{i+1} + \cdots + \Delta y_{k'}$.
    By adding $\sum_{1 \leq j \leq i} \Delta y_j$ to both sides and using \cref{eq:modular}, we obtain the following:
    $
        \overline{f_R}(Y_i \uplus X_{k'-i}) = \overline{f_R}(Y_i) + \overline{f_R}(X_{k'-i}) \geq (1-\kappa)\overline{f_R}(Y').
    $
\end{itemize}
We thus have that, in either case, 
\begin{align}
\label{eq:curvature-ineq}
    \overline{f_R}(X_{k'-i} \uplus Y_i) \geq (1-\kappa) \min\{\overline{f_R}(X'), \overline{f_R}(Y')\}.
\end{align}
Observing that \cref{eq:curvature-ineq} is true even if $i=0,k'$,
we bound the value $f(\calS)$ of the resulting reconfiguration sequence $\calS$ as follows:
\begin{align*}
    f(\calS) & = \min_{0 \leq i \leq k'} f_R(X_{k'-i} \uplus Y_i) + f(R) \\
    & \geq \min_{0 \leq i \leq k'} (1-\kappa)\overline{f_R}(X_{k'-i} \uplus Y_i) + f(R) \\
    & \geq (1-\kappa)^2 \min\{ \overline{f_R}(X'), \overline{f_R}(Y') \} + f(R) \\
    & \geq (1-\kappa)^2 \min\{ f(X), f(Y) \}. \qedhere
\end{align*}
\end{proof}

\begin{proof}[Proof of \cref{obs:msr-3-4}]
We explicitly construct such an instance that meets the specification.
Define $n \triangleq 5$, $\Sigma \triangleq \{a,b,c,d\}$,
$V_1 = \{a,b\}$, $V_2 = \{c,d\}$, $V_3 = \{a,c\}$, $V_4 = \{b,d\}$, and $V_5 = \Sigma$.
We then define a coverage function $f: 2^{[n]} \to \bbR_+$ such that
$
    f(S) \triangleq
    \frac{| \bigcup_{i \in S} V_i |}{ |\Sigma|}
$ for $S \subseteq [n]$.
Consider \MMSR defined by
$f$, $X \triangleq \{1,2\}$, and $Y \triangleq \{3,4\}$.
An optimal reconfiguration sequence from $X$ to $Y$ is
$\calS^* = \langle \{1,2\}, \{1,5\}, \{3,5\}, \{3,4\} \rangle$,
whose value is $f(\calS^*) = 1$.
On the other hand, when we are restricted to have subsets of
$X \cup Y = \{1,2,3,4\}$ in the output sequence,
we must touch either of
$\{1,3\}, \{1,4\}, \{2,3\}, \{2,4\}$, whose function value is $\frac{3}{4}$.
\end{proof}

\begin{proof}[Proof of \cref{obs:musrtjar}]
We construct such an instance that meets the specification.
Suppose $n$ is a positive integer divisible by $4$.
Define an edge-weighted graph $G=([n], E)$, where
$E \triangleq \{ (i, \frac{n}{2}+i) \mid i \in [\frac{n}{2}] \}$, and
the weight of edge $(i, \frac{n}{2}+i)$ is $i^{-1}$.
Let $f:2^{[n]} \to \bbR_+$ be a weighted cut function defined by $G$.
Consider \MUSRTJAR defined by $f$,
$X \triangleq [\frac{n}{2}]$, and $Y \triangleq [n] \setminus S$.
\cref{alg:usr} produces a reconfiguration sequence of value $1$.
On the other hand, \cref{alg:swap} returns a reconfiguration sequence includes
$\{1, \ldots, \frac{n}{4}\} \uplus \{\frac{n}{2}+1, \ldots, \frac{n}{2}+\frac{n}{4}\}$, whose cut value is $0$.
\end{proof}

\begin{proof}[Proof of \cref{obs:usr-tar-approx}]
We construct such an instance that meets the specification.
We define a submodular set function $f: 2^{[2]} \to \bbR_+$ as
$f(\emptyset) = f(\{1,2\}) = 0$ and $f(\{1\}) = f(\{2\}) = 1$.
Consider \MUSRTAR defined by $f$, $X \triangleq \{1\}$, and $Y \triangleq \{2\}$.
Since we can use \TAR steps only,
any reconfiguration sequence $\calS$ from $X$ to $Y$ (including the optimal one) must pass through at least either one of
$\emptyset$ or $\{1,2\}$; thus, $f(\calS)$ must be $0$.
\end{proof}

\begin{proof}[Proof of \cref{thm:msr-inapprox}]
We show that to solve \MVCR \emph{exactly}, a $(1-\frac{1+\epsilon}{n^2})$-approximation algorithm for \MMSR is sufficient.
Recall that given a graph $G=(V,E)$ and two minimum vertex covers $C^x$ and $C^y$,
the polynomial-time reduction introduced in the proof of \cref{thm:msr-pspace} constructs
an instance $f, C^x, C^y$ of \MMSR, for which
an optimal reconfiguration sequence $\calS^*$ satisfies that
$f(\calS^*) = |E|$ if the answer to the \MVCR instance is ``yes'' and
$f(\calS^*) \leq |E|-1$ otherwise.
To distinguish the two cases,
it is sufficient to approximate the optimal value of \MMSR within a factor of
$\frac{|E|-1+\epsilon}{|E|} \leq 1 - \frac{1+\epsilon}{n^2}$ for any $\epsilon > 0$,
completing the proof.
\end{proof}

\begin{proof}[Proof of \cref{thm:usr-tar-inapprox}]
The proof is almost the same as that of \cref{thm:usr-tjar-inapprox}.
We only need to claim that
for any $S \subseteq [n]$,
we can construct a reconfiguration sequence $\calS$ whose value is 
$g(\calS) = \Upsilon + f(S)$ using \emph{only} \TAR steps:
an example of such a sequence is
$\langle X, \cdots $ adding elements of $S$ one by one $\cdots, X \cup S, \{x_1\} \cup S, \{x_1,y_1\} \cup S, \{y_1\} \cup S, Y \cup S, \cdots $
removing elements of $S$ one by one $\cdots, Y \rangle$.
\end{proof}

\section{A* Search Algorithm}
\label{app:astar}

\cref{alg:astar} describes an A* search algorithm for \MSR and \USR.
In A* algorithms \cite{hart1968formal},
we have a table $g$ for storing
the minimum number of reconfiguration steps required to transform from $X$ to each set $S$, and
a heuristic function $h: 2^{[n]} \to \bbR$ for underestimating
the number of reconfiguration steps required to transform from each set $S$ to $Y$.
For example, $h(S) = \frac{|S \setminus Y| + |Y \setminus S|}{2}$ under \TJ and
$h(S) = \max(|S \setminus Y|, |Y \setminus S|)$ under \TJAR,
which are both \emph{admissible} and \emph{consistent} \cite{pearl1984heuristics}.
We continue the iterations,
which pop a set $S$ with minimum $g[S] + h(S)$ and
explore each of the adjacent feasible sets $T$,
until we found $T=Y$ or no further expansion is possible.
Since two or more tie sets may have the same score,
we used a Last-In-First-Out policy \cite{asai2016tiebreaking}.
Note that \cref{alg:astar} may require exponential time in the worst case.
\input{alg-astar}

\section{Additional Result}
\label{app:table}

\begin{table}[htbp]
\centering
\caption{Sequence returned by \cref{alg:swap} on \karate.}
\label{tab:my_label}
\scriptsize
\begin{tabular}{c|l|c}
\toprule
index $i$ & seed set $S^{(i)}$ & influence $f(S^{(i)})$ \\
\midrule
0 & $\{34, 2, 7, 32, 4, 11, 30, 9\}$ & 23.2 \\
1 & $\{34, 2, 7, 32, 4, 11, 30, 1\}$ & 24.3 \\
2 & $\{34, 2, 7, 32, 4, 11, 33, 1\}$ & 25.2 \\
3 & $\{34, 2, 7, 32, 4, 3, 33, 1\}$ & 25.5 \\
4 & $\{34, 2, 7, 32, 6, 3, 33, 1\}$ & 25.5 \\
5 & $\{34, 2, 7, 14, 6, 3, 33, 1\}$ & 25.0 \\
6 & $\{34, 2, 24, 14, 6, 3, 33, 1\}$ & 25.2 \\
7 & $\{34, 25, 24, 14, 6, 3, 33, 1\}$ & 25.0 \\
8 & $\{27, 25, 24, 14, 6, 3, 33, 1\}$ & 23.6 \\
\bottomrule
\end{tabular}
\end{table}

%% file: alg-astar.tex
\begin{algorithm}[tbp]
\caption{A* algorithm for \MSR and \USR.}
\label{alg:astar}
\footnotesize
\begin{algorithmic}[1]
    \Require
    submodular function $f: 2^{[n]} \to \bbR_+$,
    two sets $X, Y$,
    threshold $\theta$,
    heuristic function $h: 2^{[n]} \to \bbR_+$.
    \State initialize priority queue $\OPEN$, and $\CLOSE \leftarrow \emptyset$.
    \State declare empty hash table $g$, and push $X$ with score $h(X)$ into $\OPEN$.
    \While{$\OPEN$ is not empty}
        \State pop set $S$ with minimum score from $\OPEN$ \textbf{and} push $S$ into $\CLOSE$.
        \State \textbf{if} $S = Y$ \textbf{then}:
        \textbf{return} reconfiguration sequence constructed using $\pi$.
        \ForAll{set $T$ adjacent to $S$ such that $f(T) \geq \theta$}
            \If{$T \in \OPEN$ \textbf{and} $g[S] + 1 < g[T]$} \Comment{reinsert node.}
                    \State $g[T] \leftarrow g[S]+1$ \textbf{and} $\pi[T] \leftarrow S$.
                    \State remove $T$ from $\OPEN$ \textbf{and} push $T$ with score $g[T] + h(T)$ into $\OPEN$.
            \ElsIf{$T \in \CLOSE$ \textbf{and} $g[S] + 1 < g[T]$} \Comment{reopen node.}
                    \State $g[T] \leftarrow g[S] + 1$ \textbf{and} $\pi[T] \leftarrow S$.
                    \State remove $T$ from $\CLOSE$ \textbf{and} push $T$ with score $g[T] + h(T)$ into $\OPEN$.
            \ElsIf{$T \not\in \OPEN$ \textbf{and} $T \not\in \CLOSE$} \Comment{open node.}
                \State $g[T] \leftarrow g[S] + 1$ \textbf{and} $\pi[T] \leftarrow S $.
                \State push $T$ with score $g[T]+h(T)$ into $\OPEN$.
            \EndIf
        \EndFor
    \EndWhile
\end{algorithmic}
\end{algorithm}

%% file: head.bbl

\begin{thebibliography}{64}


\ifx \showCODEN    \undefined \def \showCODEN     #1{\unskip}     \fi
\ifx \showDOI      \undefined \def \showDOI       #1{#1}\fi
\ifx \showISBNx    \undefined \def \showISBNx     #1{\unskip}     \fi
\ifx \showISBNxiii \undefined \def \showISBNxiii  #1{\unskip}     \fi
\ifx \showISSN     \undefined \def \showISSN      #1{\unskip}     \fi
\ifx \showLCCN     \undefined \def \showLCCN      #1{\unskip}     \fi
\ifx \shownote     \undefined \def \shownote      #1{#1}          \fi
\ifx \showarticletitle \undefined \def \showarticletitle #1{#1}   \fi
\ifx \showURL      \undefined \def \showURL       {\relax}        \fi
\providecommand\bibfield[2]{#2}
\providecommand\bibinfo[2]{#2}
\providecommand\natexlab[1]{#1}
\providecommand\showeprint[2][]{arXiv:#2}

\bibitem[\protect\citeauthoryear{Adomavicius and Zhang}{Adomavicius and
  Zhang}{2012}]%
        {adomavicius2012stability}
\bibfield{author}{\bibinfo{person}{Gediminas Adomavicius} {and}
  \bibinfo{person}{Jingjing Zhang}.} \bibinfo{year}{2012}\natexlab{}.
\newblock \showarticletitle{Stability of recommendation algorithms}.
\newblock \bibinfo{journal}{\emph{{ACM} Trans. Inf. Syst.}}
  \bibinfo{volume}{30}, \bibinfo{number}{4} (\bibinfo{year}{2012}),
  \bibinfo{pages}{23:1--23:31}.
\newblock


\bibitem[\protect\citeauthoryear{Arora, Galhotra, and Ranu}{Arora
  et~al\mbox{.}}{2017}]%
        {arora2017debunking}
\bibfield{author}{\bibinfo{person}{Akhil Arora}, \bibinfo{person}{Sainyam
  Galhotra}, {and} \bibinfo{person}{Sayan Ranu}.}
  \bibinfo{year}{2017}\natexlab{}.
\newblock \showarticletitle{Debunking the Myths of Influence Maximization: An
  In-Depth Benchmarking Study}. In \bibinfo{booktitle}{\emph{SIGMOD}}.
  \bibinfo{pages}{651--666}.
\newblock


\bibitem[\protect\citeauthoryear{Arora and Barak}{Arora and Barak}{2009}]%
        {arora2009computational}
\bibfield{author}{\bibinfo{person}{Sanjeev Arora} {and} \bibinfo{person}{Boaz
  Barak}.} \bibinfo{year}{2009}\natexlab{}.
\newblock \bibinfo{booktitle}{\emph{Computational Complexity: A Modern
  Approach}}.
\newblock \bibinfo{publisher}{Cambridge University Press}.
\newblock


\bibitem[\protect\citeauthoryear{Asai and Fukunaga}{Asai and Fukunaga}{2016}]%
        {asai2016tiebreaking}
\bibfield{author}{\bibinfo{person}{Masataro Asai} {and} \bibinfo{person}{Alex
  Fukunaga}.} \bibinfo{year}{2016}\natexlab{}.
\newblock \showarticletitle{Tiebreaking strategies for {A*} search: How to
  explore the final frontier}. In \bibinfo{booktitle}{\emph{AAAI}}.
  \bibinfo{pages}{673--679}.
\newblock


\bibitem[\protect\citeauthoryear{Bonsma and Cereceda}{Bonsma and
  Cereceda}{2009}]%
        {bonsma2009finding}
\bibfield{author}{\bibinfo{person}{Paul Bonsma} {and} \bibinfo{person}{Luis
  Cereceda}.} \bibinfo{year}{2009}\natexlab{}.
\newblock \showarticletitle{Finding paths between graph colourings:
  {PSPACE}-completeness and superpolynomial distances}.
\newblock \bibinfo{journal}{\emph{Theor. Comput. Sci.}} \bibinfo{volume}{410},
  \bibinfo{number}{50} (\bibinfo{year}{2009}), \bibinfo{pages}{5215--5226}.
\newblock


\bibitem[\protect\citeauthoryear{Borgs, Brautbar, Chayes, and Lucier}{Borgs
  et~al\mbox{.}}{2014}]%
        {borgs2014maximizing}
\bibfield{author}{\bibinfo{person}{Christian Borgs}, \bibinfo{person}{Michael
  Brautbar}, \bibinfo{person}{Jennifer Chayes}, {and} \bibinfo{person}{Brendan
  Lucier}.} \bibinfo{year}{2014}\natexlab{}.
\newblock \showarticletitle{Maximizing Social Influence in Nearly Optimal
  Time}. In \bibinfo{booktitle}{\emph{SODA}}. \bibinfo{pages}{946--957}.
\newblock


\bibitem[\protect\citeauthoryear{Borodin and Rains}{Borodin and Rains}{2005}]%
        {borodin2005eynard}
\bibfield{author}{\bibinfo{person}{Alexei Borodin} {and}
  \bibinfo{person}{Eric~M. Rains}.} \bibinfo{year}{2005}\natexlab{}.
\newblock \showarticletitle{{Eynard-Mehta} theorem, {Schur} process, and their
  {Pfaffian} analogs}.
\newblock \bibinfo{journal}{\emph{J. Stat. Phys.}} \bibinfo{volume}{121},
  \bibinfo{number}{3--4} (\bibinfo{year}{2005}), \bibinfo{pages}{291--317}.
\newblock


\bibitem[\protect\citeauthoryear{Boutsidis and Gallopoulos}{Boutsidis and
  Gallopoulos}{2008}]%
        {boutsidis2008svd}
\bibfield{author}{\bibinfo{person}{Christos Boutsidis} {and}
  \bibinfo{person}{Efstratios Gallopoulos}.} \bibinfo{year}{2008}\natexlab{}.
\newblock \showarticletitle{{SVD} based initialization: A head start for
  nonnegative matrix factorization}.
\newblock \bibinfo{journal}{\emph{Pattern Recognit.}} \bibinfo{volume}{41},
  \bibinfo{number}{4} (\bibinfo{year}{2008}), \bibinfo{pages}{1350--1362}.
\newblock


\bibitem[\protect\citeauthoryear{Buchbinder and Feldman}{Buchbinder and
  Feldman}{2018a}]%
        {buchbinder2018deterministic}
\bibfield{author}{\bibinfo{person}{Niv Buchbinder} {and} \bibinfo{person}{Moran
  Feldman}.} \bibinfo{year}{2018}\natexlab{a}.
\newblock \showarticletitle{Deterministic algorithms for submodular
  maximization problems}.
\newblock \bibinfo{journal}{\emph{ACM Trans. Algorithms}} \bibinfo{volume}{14},
  \bibinfo{number}{3} (\bibinfo{year}{2018}), \bibinfo{pages}{1--20}.
\newblock


\bibitem[\protect\citeauthoryear{Buchbinder and Feldman}{Buchbinder and
  Feldman}{2018b}]%
        {buchbinder2018submodular}
\bibfield{author}{\bibinfo{person}{Niv Buchbinder} {and} \bibinfo{person}{Moran
  Feldman}.} \bibinfo{year}{2018}\natexlab{b}.
\newblock \showarticletitle{Submodular Functions Maximization Problems}.
\newblock In \bibinfo{booktitle}{\emph{Handbook of Approximation Algorithms and
  Metaheuristics}}. \bibinfo{publisher}{Chapman and Hall/CRC},
  \bibinfo{pages}{771--806}.
\newblock


\bibitem[\protect\citeauthoryear{Buchbinder, Feldman, Seffi, and
  Schwartz}{Buchbinder et~al\mbox{.}}{2015}]%
        {buchbinder2015tight}
\bibfield{author}{\bibinfo{person}{Niv Buchbinder}, \bibinfo{person}{Moran
  Feldman}, \bibinfo{person}{Joseph Seffi}, {and} \bibinfo{person}{Roy
  Schwartz}.} \bibinfo{year}{2015}\natexlab{}.
\newblock \showarticletitle{A tight linear time (1/2)-approximation for
  unconstrained submodular maximization}.
\newblock \bibinfo{journal}{\emph{SIAM J. Comput.}} \bibinfo{volume}{44},
  \bibinfo{number}{5} (\bibinfo{year}{2015}), \bibinfo{pages}{1384--1402}.
\newblock


\bibitem[\protect\citeauthoryear{Cardinal, Demaine, Eppstein, Hearn, and
  Winslow}{Cardinal et~al\mbox{.}}{2020}]%
        {cardinal2020reconfiguration}
\bibfield{author}{\bibinfo{person}{Jean Cardinal}, \bibinfo{person}{Erik~D.
  Demaine}, \bibinfo{person}{David Eppstein}, \bibinfo{person}{Robert~A.
  Hearn}, {and} \bibinfo{person}{Andrew Winslow}.}
  \bibinfo{year}{2020}\natexlab{}.
\newblock \showarticletitle{Reconfiguration of satisfying assignments and
  subset sums: Easy to find, hard to connect}.
\newblock \bibinfo{journal}{\emph{Theor. Comput. Sci.}}  \bibinfo{volume}{806}
  (\bibinfo{year}{2020}), \bibinfo{pages}{332--343}.
\newblock


\bibitem[\protect\citeauthoryear{Chen, Zhang, and Zhou}{Chen
  et~al\mbox{.}}{2018}]%
        {chen2018fast}
\bibfield{author}{\bibinfo{person}{Laming Chen}, \bibinfo{person}{Guoxin
  Zhang}, {and} \bibinfo{person}{Eric Zhou}.} \bibinfo{year}{2018}\natexlab{}.
\newblock \showarticletitle{Fast greedy {MAP} inference for determinantal point
  process to improve recommendation diversity}. In
  \bibinfo{booktitle}{\emph{NeurIPS}}. \bibinfo{pages}{5622--5633}.
\newblock


\bibitem[\protect\citeauthoryear{Chen, Wang, and Wang}{Chen
  et~al\mbox{.}}{2010}]%
        {chen2010scalable}
\bibfield{author}{\bibinfo{person}{Wei Chen}, \bibinfo{person}{Chi Wang}, {and}
  \bibinfo{person}{Yajun Wang}.} \bibinfo{year}{2010}\natexlab{}.
\newblock \showarticletitle{Scalable Influence Maximization for Prevalent Viral
  Marketing in Large-Scale Social Networks}. In
  \bibinfo{booktitle}{\emph{KDD}}. \bibinfo{pages}{1029--1038}.
\newblock


\bibitem[\protect\citeauthoryear{Conforti and Cornu{\'e}jols}{Conforti and
  Cornu{\'e}jols}{1984}]%
        {conforti1984submodular}
\bibfield{author}{\bibinfo{person}{Michele Conforti} {and}
  \bibinfo{person}{G{\'e}rard Cornu{\'e}jols}.}
  \bibinfo{year}{1984}\natexlab{}.
\newblock \showarticletitle{Submodular set functions, matroids and the greedy
  algorithm: Tight worst-case bounds and some generalizations of the
  {Rado}-{Edmonds} theorem}.
\newblock \bibinfo{journal}{\emph{Discrete Appl. Math.}} \bibinfo{volume}{7},
  \bibinfo{number}{3} (\bibinfo{year}{1984}), \bibinfo{pages}{251--274}.
\newblock


\bibitem[\protect\citeauthoryear{Cook}{Cook}{1971}]%
        {cook1971complexity}
\bibfield{author}{\bibinfo{person}{Stephen~A. Cook}.}
  \bibinfo{year}{1971}\natexlab{}.
\newblock \showarticletitle{The complexity of theorem-proving procedures}. In
  \bibinfo{booktitle}{\emph{STOC}}. \bibinfo{pages}{151--158}.
\newblock


\bibitem[\protect\citeauthoryear{Domingos and Richardson}{Domingos and
  Richardson}{2001}]%
        {domingos2001mining}
\bibfield{author}{\bibinfo{person}{Pedro Domingos} {and} \bibinfo{person}{Matt
  Richardson}.} \bibinfo{year}{2001}\natexlab{}.
\newblock \showarticletitle{Mining the Network Value of Customers}. In
  \bibinfo{booktitle}{\emph{KDD}}. \bibinfo{pages}{57--66}.
\newblock


\bibitem[\protect\citeauthoryear{Feige}{Feige}{1998}]%
        {feige1998threshold}
\bibfield{author}{\bibinfo{person}{Uriel Feige}.}
  \bibinfo{year}{1998}\natexlab{}.
\newblock \showarticletitle{A threshold of ln $n$ for approximating set cover}.
\newblock \bibinfo{journal}{\emph{J. ACM}} \bibinfo{volume}{45},
  \bibinfo{number}{4} (\bibinfo{year}{1998}), \bibinfo{pages}{634--652}.
\newblock


\bibitem[\protect\citeauthoryear{Feige, Mirrokni, and Vondr{\'a}k}{Feige
  et~al\mbox{.}}{2011}]%
        {feige2011maximizing}
\bibfield{author}{\bibinfo{person}{Uriel Feige}, \bibinfo{person}{Vahab~S.
  Mirrokni}, {and} \bibinfo{person}{Jan Vondr{\'a}k}.}
  \bibinfo{year}{2011}\natexlab{}.
\newblock \showarticletitle{Maximizing non-monotone submodular functions}.
\newblock \bibinfo{journal}{\emph{SIAM J. Comput.}} \bibinfo{volume}{40},
  \bibinfo{number}{4} (\bibinfo{year}{2011}), \bibinfo{pages}{1133--1153}.
\newblock


\bibitem[\protect\citeauthoryear{Garey and Johnson}{Garey and Johnson}{1979}]%
        {garey1979computers}
\bibfield{author}{\bibinfo{person}{Michael~R. Garey} {and}
  \bibinfo{person}{David~S. Johnson}.} \bibinfo{year}{1979}\natexlab{}.
\newblock \bibinfo{booktitle}{\emph{Computers and Intractability: A Guide to
  the Theory of {NP}-Completeness}}.
\newblock \bibinfo{publisher}{W. H. Freeman}.
\newblock


\bibitem[\protect\citeauthoryear{Gillenwater, Kulesza, and Taskar}{Gillenwater
  et~al\mbox{.}}{2012}]%
        {gillenwater2012near}
\bibfield{author}{\bibinfo{person}{Jennifer Gillenwater}, \bibinfo{person}{Alex
  Kulesza}, {and} \bibinfo{person}{Ben Taskar}.}
  \bibinfo{year}{2012}\natexlab{}.
\newblock \showarticletitle{Near-optimal {MAP} inference for determinantal
  point processes}. In \bibinfo{booktitle}{\emph{NIPS}}.
  \bibinfo{pages}{2735--2743}.
\newblock


\bibitem[\protect\citeauthoryear{Goldenberg, Libai, and Muller}{Goldenberg
  et~al\mbox{.}}{2001}]%
        {goldenberg2001talk}
\bibfield{author}{\bibinfo{person}{Jacob Goldenberg}, \bibinfo{person}{Barak
  Libai}, {and} \bibinfo{person}{Eitan Muller}.}
  \bibinfo{year}{2001}\natexlab{}.
\newblock \showarticletitle{Talk of the network: A complex systems look at the
  underlying process of word-of-mouth}.
\newblock \bibinfo{journal}{\emph{Mark. Lett.}} \bibinfo{volume}{12},
  \bibinfo{number}{3} (\bibinfo{year}{2001}), \bibinfo{pages}{211--223}.
\newblock


\bibitem[\protect\citeauthoryear{Gopalan, Kolaitis, Maneva, and
  Papadimitriou}{Gopalan et~al\mbox{.}}{2009}]%
        {gopalan2009connectivity}
\bibfield{author}{\bibinfo{person}{Parikshit Gopalan},
  \bibinfo{person}{Phokion~G. Kolaitis}, \bibinfo{person}{Elitza Maneva}, {and}
  \bibinfo{person}{Christos~H. Papadimitriou}.}
  \bibinfo{year}{2009}\natexlab{}.
\newblock \showarticletitle{The connectivity of Boolean satisfiability:
  computational and structural dichotomies}.
\newblock \bibinfo{journal}{\emph{SIAM J. Comput.}} \bibinfo{volume}{38},
  \bibinfo{number}{6} (\bibinfo{year}{2009}), \bibinfo{pages}{2330--2355}.
\newblock


\bibitem[\protect\citeauthoryear{Harper and Konstan}{Harper and
  Konstan}{2015}]%
        {harper2015movielens}
\bibfield{author}{\bibinfo{person}{F.~Maxwell Harper} {and}
  \bibinfo{person}{Joseph~A. Konstan}.} \bibinfo{year}{2015}\natexlab{}.
\newblock \showarticletitle{The {MovieLens} datasets: History and context}.
\newblock \bibinfo{journal}{\emph{ACM Trans. Interact. Intell. Syst.}}
  \bibinfo{volume}{5}, \bibinfo{number}{4} (\bibinfo{year}{2015}),
  \bibinfo{pages}{1--19}.
\newblock


\bibitem[\protect\citeauthoryear{Hart, Nilsson, and Raphael}{Hart
  et~al\mbox{.}}{1968}]%
        {hart1968formal}
\bibfield{author}{\bibinfo{person}{Peter~E. Hart}, \bibinfo{person}{Nils~J.
  Nilsson}, {and} \bibinfo{person}{Bertram Raphael}.}
  \bibinfo{year}{1968}\natexlab{}.
\newblock \showarticletitle{A formal basis for the heuristic determination of
  minimum cost paths}.
\newblock \bibinfo{journal}{\emph{IEEE Trans. Syst. Sci. Cybern.}}
  \bibinfo{volume}{4}, \bibinfo{number}{2} (\bibinfo{year}{1968}),
  \bibinfo{pages}{100--107}.
\newblock


\bibitem[\protect\citeauthoryear{Hearn and Demaine}{Hearn and Demaine}{2005}]%
        {hearn2005pspace}
\bibfield{author}{\bibinfo{person}{Robert~A. Hearn} {and}
  \bibinfo{person}{Erik~D. Demaine}.} \bibinfo{year}{2005}\natexlab{}.
\newblock \showarticletitle{{PSPACE}-Completeness of Sliding-Block Puzzles and
  Other Problems through the Nondeterministic Constraint Logic Model of
  Computation}.
\newblock \bibinfo{journal}{\emph{Theor. Comput. Sci.}} \bibinfo{volume}{343},
  \bibinfo{number}{1-2} (\bibinfo{year}{2005}), \bibinfo{pages}{72--96}.
\newblock


\bibitem[\protect\citeauthoryear{Ito and Demaine}{Ito and Demaine}{2014}]%
        {ito2014approximability}
\bibfield{author}{\bibinfo{person}{Takehiro Ito} {and} \bibinfo{person}{Erik~D.
  Demaine}.} \bibinfo{year}{2014}\natexlab{}.
\newblock \showarticletitle{Approximability of the subset sum reconfiguration
  problem}.
\newblock \bibinfo{journal}{\emph{J. Comb. Optim.}} \bibinfo{volume}{28},
  \bibinfo{number}{3} (\bibinfo{year}{2014}), \bibinfo{pages}{639--654}.
\newblock


\bibitem[\protect\citeauthoryear{Ito, Demaine, Harvey, Papadimitriou, Sideri,
  Uehara, and Uno}{Ito et~al\mbox{.}}{2011}]%
        {ito2011complexity}
\bibfield{author}{\bibinfo{person}{Takehiro Ito}, \bibinfo{person}{Erik~D.
  Demaine}, \bibinfo{person}{Nicholas~J.A. Harvey},
  \bibinfo{person}{Christos~H. Papadimitriou}, \bibinfo{person}{Martha Sideri},
  \bibinfo{person}{Ryuhei Uehara}, {and} \bibinfo{person}{Yushi Uno}.}
  \bibinfo{year}{2011}\natexlab{}.
\newblock \showarticletitle{On the complexity of reconfiguration problems}.
\newblock \bibinfo{journal}{\emph{Theor. Comput. Sci.}} \bibinfo{volume}{412},
  \bibinfo{number}{12-14} (\bibinfo{year}{2011}), \bibinfo{pages}{1054--1065}.
\newblock


\bibitem[\protect\citeauthoryear{Ito, Nooka, and Zhou}{Ito
  et~al\mbox{.}}{2016}]%
        {ito2016reconfiguration}
\bibfield{author}{\bibinfo{person}{Takehiro Ito}, \bibinfo{person}{Hiroyuki
  Nooka}, {and} \bibinfo{person}{Xiao Zhou}.} \bibinfo{year}{2016}\natexlab{}.
\newblock \showarticletitle{Reconfiguration of vertex covers in a graph}.
\newblock \bibinfo{journal}{\emph{IEICE Trans. Inf. \& Syst.}}
  \bibinfo{volume}{99}, \bibinfo{number}{3} (\bibinfo{year}{2016}),
  \bibinfo{pages}{598--606}.
\newblock


\bibitem[\protect\citeauthoryear{Iyer, Jegelka, and Bilmes}{Iyer
  et~al\mbox{.}}{2013}]%
        {iyer2013curvature}
\bibfield{author}{\bibinfo{person}{Rishabh~K. Iyer}, \bibinfo{person}{Stefanie
  Jegelka}, {and} \bibinfo{person}{Jeff~A. Bilmes}.}
  \bibinfo{year}{2013}\natexlab{}.
\newblock \showarticletitle{Curvature and optimal algorithms for learning and
  minimizing submodular functions}. In \bibinfo{booktitle}{\emph{NIPS}}.
  \bibinfo{pages}{2742--2750}.
\newblock


\bibitem[\protect\citeauthoryear{Johnson, Kratsch, Kratsch, Patel, and
  Paulusma}{Johnson et~al\mbox{.}}{2016}]%
        {johnson2016finding}
\bibfield{author}{\bibinfo{person}{Matthew Johnson}, \bibinfo{person}{Dieter
  Kratsch}, \bibinfo{person}{Stefan Kratsch}, \bibinfo{person}{Viresh Patel},
  {and} \bibinfo{person}{Dani{\"e}l Paulusma}.}
  \bibinfo{year}{2016}\natexlab{}.
\newblock \showarticletitle{Finding shortest paths between graph colourings}.
\newblock \bibinfo{journal}{\emph{Algorithmica}} \bibinfo{volume}{75},
  \bibinfo{number}{2} (\bibinfo{year}{2016}), \bibinfo{pages}{295--321}.
\newblock


\bibitem[\protect\citeauthoryear{Johnson and Story}{Johnson and Story}{1879}]%
        {johnson1879notes}
\bibfield{author}{\bibinfo{person}{Wm~Woolsey Johnson} {and}
  \bibinfo{person}{William~Edward Story}.} \bibinfo{year}{1879}\natexlab{}.
\newblock \showarticletitle{Notes on the ``15'' puzzle}.
\newblock \bibinfo{journal}{\emph{Am. J. Math.}} \bibinfo{volume}{2},
  \bibinfo{number}{4} (\bibinfo{year}{1879}), \bibinfo{pages}{397--404}.
\newblock


\bibitem[\protect\citeauthoryear{Kami{\'n}ski, Medvedev, and
  Milani{\v{c}}}{Kami{\'n}ski et~al\mbox{.}}{2012}]%
        {kaminski2012complexity}
\bibfield{author}{\bibinfo{person}{Marcin Kami{\'n}ski}, \bibinfo{person}{Paul
  Medvedev}, {and} \bibinfo{person}{Martin Milani{\v{c}}}.}
  \bibinfo{year}{2012}\natexlab{}.
\newblock \showarticletitle{Complexity of independent set reconfigurability
  problems}.
\newblock \bibinfo{journal}{\emph{Theor. Comput. Sci.}}  \bibinfo{volume}{439}
  (\bibinfo{year}{2012}), \bibinfo{pages}{9--15}.
\newblock


\bibitem[\protect\citeauthoryear{Karp}{Karp}{1972}]%
        {karp1972reducibility}
\bibfield{author}{\bibinfo{person}{Richard~M. Karp}.}
  \bibinfo{year}{1972}\natexlab{}.
\newblock \showarticletitle{Reducibility among combinatorial problems}. In
  \bibinfo{booktitle}{\emph{Complexity of Computer Computations}}.
  \bibinfo{pages}{85--103}.
\newblock


\bibitem[\protect\citeauthoryear{Kempe, Kleinberg, and Tardos}{Kempe
  et~al\mbox{.}}{2003}]%
        {kempe2003maximizing}
\bibfield{author}{\bibinfo{person}{David Kempe}, \bibinfo{person}{Jon
  Kleinberg}, {and} \bibinfo{person}{{\'{E}}va Tardos}.}
  \bibinfo{year}{2003}\natexlab{}.
\newblock \showarticletitle{Maximizing the Spread of Influence through a Social
  Network}. In \bibinfo{booktitle}{\emph{KDD}}. \bibinfo{pages}{137--146}.
\newblock


\bibitem[\protect\citeauthoryear{Krause and Golovin}{Krause and
  Golovin}{2014}]%
        {krause2014submodular}
\bibfield{author}{\bibinfo{person}{Andreas Krause} {and}
  \bibinfo{person}{Daniel Golovin}.} \bibinfo{year}{2014}\natexlab{}.
\newblock \showarticletitle{Submodular Function Maximization}.
\newblock In \bibinfo{booktitle}{\emph{Tractability: Practical Approaches to
  Hard Problems}}. \bibinfo{pages}{71--104}.
\newblock


\bibitem[\protect\citeauthoryear{Krause, Singh, and Guestrin}{Krause
  et~al\mbox{.}}{2008}]%
        {krause2008near}
\bibfield{author}{\bibinfo{person}{Andreas Krause}, \bibinfo{person}{Ajit~Paul
  Singh}, {and} \bibinfo{person}{Carlos Guestrin}.}
  \bibinfo{year}{2008}\natexlab{}.
\newblock \showarticletitle{Near-optimal sensor placements in Gaussian
  processes: Theory, efficient algorithms and empirical studies}.
\newblock \bibinfo{journal}{\emph{J. Mach. Learn. Res.}}  \bibinfo{volume}{9}
  (\bibinfo{year}{2008}), \bibinfo{pages}{235--284}.
\newblock


\bibitem[\protect\citeauthoryear{Kulesza and Taskar}{Kulesza and
  Taskar}{2012}]%
        {kulesza2012determinantal}
\bibfield{author}{\bibinfo{person}{Alex Kulesza} {and} \bibinfo{person}{Ben
  Taskar}.} \bibinfo{year}{2012}\natexlab{}.
\newblock \showarticletitle{Determinantal Point Processes for Machine
  Learning}.
\newblock \bibinfo{journal}{\emph{Found. Trends Mach. Learn.}}
  \bibinfo{volume}{5}, \bibinfo{number}{2--3} (\bibinfo{year}{2012}),
  \bibinfo{pages}{123--286}.
\newblock


\bibitem[\protect\citeauthoryear{Kunegis}{Kunegis}{2013}]%
        {kunegis2013konect}
\bibfield{author}{\bibinfo{person}{J\'{e}r\^{o}me Kunegis}.}
  \bibinfo{year}{2013}\natexlab{}.
\newblock \showarticletitle{{KONECT} -- {The} {Koblenz} {Network}
  {Collection}}. In \bibinfo{booktitle}{\emph{WWW Companion}}.
  \bibinfo{pages}{1343--1350}.
\newblock


\bibitem[\protect\citeauthoryear{Leskovec, Kleinberg, and Faloutsos}{Leskovec
  et~al\mbox{.}}{2007a}]%
        {leskovec2007graph}
\bibfield{author}{\bibinfo{person}{Jure Leskovec}, \bibinfo{person}{Jon
  Kleinberg}, {and} \bibinfo{person}{Christos Faloutsos}.}
  \bibinfo{year}{2007}\natexlab{a}.
\newblock \showarticletitle{Graph Evolution: Densification and Shrinking
  Diameters}.
\newblock \bibinfo{journal}{\emph{ACM Trans. Knowl. Discov. Data}}
  \bibinfo{volume}{1}, \bibinfo{number}{1} (\bibinfo{year}{2007}),
  \bibinfo{pages}{2}.
\newblock


\bibitem[\protect\citeauthoryear{Leskovec, Krause, Guestrin, Faloutsos,
  VanBriesen, and Glance}{Leskovec et~al\mbox{.}}{2007b}]%
        {leskovec2007cost}
\bibfield{author}{\bibinfo{person}{Jure Leskovec}, \bibinfo{person}{Andreas
  Krause}, \bibinfo{person}{Carlos Guestrin}, \bibinfo{person}{Christos
  Faloutsos}, \bibinfo{person}{Jeanne VanBriesen}, {and}
  \bibinfo{person}{Natalie Glance}.} \bibinfo{year}{2007}\natexlab{b}.
\newblock \showarticletitle{Cost-effective Outbreak Detection in Networks}. In
  \bibinfo{booktitle}{\emph{KDD}}. \bibinfo{pages}{420--429}.
\newblock


\bibitem[\protect\citeauthoryear{Levin}{Levin}{1973}]%
        {levin1973universal}
\bibfield{author}{\bibinfo{person}{Leonid~Anatolevich Levin}.}
  \bibinfo{year}{1973}\natexlab{}.
\newblock \showarticletitle{Universal sequential search problems}.
\newblock \bibinfo{journal}{\emph{Probl. Peredachi Inf.}} \bibinfo{volume}{9},
  \bibinfo{number}{3} (\bibinfo{year}{1973}), \bibinfo{pages}{115--116}.
\newblock


\bibitem[\protect\citeauthoryear{Lichtenstein and Sipser}{Lichtenstein and
  Sipser}{1980}]%
        {lichtenstein1980go}
\bibfield{author}{\bibinfo{person}{David Lichtenstein} {and}
  \bibinfo{person}{Michael Sipser}.} \bibinfo{year}{1980}\natexlab{}.
\newblock \showarticletitle{Go is polynomial-space hard}.
\newblock \bibinfo{journal}{\emph{J. ACM}} \bibinfo{volume}{27},
  \bibinfo{number}{2} (\bibinfo{year}{1980}), \bibinfo{pages}{393--401}.
\newblock


\bibitem[\protect\citeauthoryear{Lin and Bilmes}{Lin and Bilmes}{2011}]%
        {lin2011class}
\bibfield{author}{\bibinfo{person}{Hui Lin} {and} \bibinfo{person}{Jeff
  Bilmes}.} \bibinfo{year}{2011}\natexlab{}.
\newblock \showarticletitle{A class of submodular functions for document
  summarization}. In \bibinfo{booktitle}{\emph{ACL-HLT}}.
  \bibinfo{pages}{510--520}.
\newblock


\bibitem[\protect\citeauthoryear{Loukides, Gwadera, and Chang}{Loukides
  et~al\mbox{.}}{2020}]%
        {loukides2020overexposure}
\bibfield{author}{\bibinfo{person}{Grigorios Loukides}, \bibinfo{person}{Robert
  Gwadera}, {and} \bibinfo{person}{Shing{-}Wan Chang}.}
  \bibinfo{year}{2020}\natexlab{}.
\newblock \showarticletitle{Overexposure-aware influence maximization}.
\newblock \bibinfo{journal}{\emph{ACM Trans. Internet Technol.}}
  \bibinfo{volume}{20}, \bibinfo{number}{4} (\bibinfo{year}{2020}),
  \bibinfo{pages}{1--31}.
\newblock


\bibitem[\protect\citeauthoryear{Lov{\'a}sz}{Lov{\'a}sz}{1983}]%
        {lovasz1983submodular}
\bibfield{author}{\bibinfo{person}{L{\'a}szl{\'o} Lov{\'a}sz}.}
  \bibinfo{year}{1983}\natexlab{}.
\newblock \showarticletitle{Submodular Functions and Convexity}.
\newblock In \bibinfo{booktitle}{\emph{Mathematical Programming -- The State of
  the Art}}. \bibinfo{pages}{235--257}.
\newblock


\bibitem[\protect\citeauthoryear{Macchi}{Macchi}{1975}]%
        {macchi1975coincidence}
\bibfield{author}{\bibinfo{person}{Odile Macchi}.}
  \bibinfo{year}{1975}\natexlab{}.
\newblock \showarticletitle{The coincidence approach to stochastic point
  processes}.
\newblock \bibinfo{journal}{\emph{Adv. Appl. Probab.}} \bibinfo{volume}{7},
  \bibinfo{number}{1} (\bibinfo{year}{1975}), \bibinfo{pages}{83--122}.
\newblock


\bibitem[\protect\citeauthoryear{Mirzasoleiman, Badanidiyuru, and
  Karbasi}{Mirzasoleiman et~al\mbox{.}}{2016}]%
        {mirzasoleiman2016fast}
\bibfield{author}{\bibinfo{person}{Baharan Mirzasoleiman},
  \bibinfo{person}{Ashwinkumar Badanidiyuru}, {and} \bibinfo{person}{Amin
  Karbasi}.} \bibinfo{year}{2016}\natexlab{}.
\newblock \showarticletitle{Fast constrained submodular maximization:
  Personalized data summarization}. In \bibinfo{booktitle}{\emph{ICML}}.
  \bibinfo{pages}{1358--1367}.
\newblock


\bibitem[\protect\citeauthoryear{Mouawad}{Mouawad}{2015}]%
        {mouawad2015reconfiguration}
\bibfield{author}{\bibinfo{person}{Amer Mouawad}.}
  \bibinfo{year}{2015}\natexlab{}.
\newblock \emph{\bibinfo{title}{On Reconfiguration Problems: Structure and
  Tractability}}.
\newblock \bibinfo{thesistype}{Ph.\,D. Dissertation}.
  \bibinfo{school}{University of Waterloo}.
\newblock


\bibitem[\protect\citeauthoryear{Nemhauser and Wolsey}{Nemhauser and
  Wolsey}{1978}]%
        {nemhauser1978best}
\bibfield{author}{\bibinfo{person}{George~L. Nemhauser} {and}
  \bibinfo{person}{Laurence~A. Wolsey}.} \bibinfo{year}{1978}\natexlab{}.
\newblock \showarticletitle{Best algorithms for approximating the maximum of a
  submodular set function}.
\newblock \bibinfo{journal}{\emph{Math. Oper. Res.}} \bibinfo{volume}{3},
  \bibinfo{number}{3} (\bibinfo{year}{1978}), \bibinfo{pages}{177--188}.
\newblock


\bibitem[\protect\citeauthoryear{Nemhauser, Wolsey, and Fisher}{Nemhauser
  et~al\mbox{.}}{1978}]%
        {nemhauser1978analysis}
\bibfield{author}{\bibinfo{person}{George~L. Nemhauser},
  \bibinfo{person}{Laurence~A. Wolsey}, {and} \bibinfo{person}{Marshall~L.
  Fisher}.} \bibinfo{year}{1978}\natexlab{}.
\newblock \showarticletitle{An analysis of the approximations for maximizing
  submodular set functions}.
\newblock \bibinfo{journal}{\emph{Math. Program.}}  \bibinfo{volume}{14}
  (\bibinfo{year}{1978}), \bibinfo{pages}{265--294}.
\newblock


\bibitem[\protect\citeauthoryear{Nishimura}{Nishimura}{2018}]%
        {nishimura2018introduction}
\bibfield{author}{\bibinfo{person}{Naomi Nishimura}.}
  \bibinfo{year}{2018}\natexlab{}.
\newblock \showarticletitle{Introduction to Reconfiguration}.
\newblock \bibinfo{journal}{\emph{Algorithms}} \bibinfo{volume}{11},
  \bibinfo{number}{4} (\bibinfo{year}{2018}), \bibinfo{pages}{52}.
\newblock


\bibitem[\protect\citeauthoryear{Ohsaka}{Ohsaka}{2020}]%
        {ohsaka2020solution}
\bibfield{author}{\bibinfo{person}{Naoto Ohsaka}.}
  \bibinfo{year}{2020}\natexlab{}.
\newblock \showarticletitle{The solution distribution of influence
  maximization: {A} high-level experimental study on three algorithmic
  approaches}. In \bibinfo{booktitle}{\emph{SIGMOD}}.
\newblock


\bibitem[\protect\citeauthoryear{Ohsaka, Akiba, Yoshida, and
  Kawarabayashi}{Ohsaka et~al\mbox{.}}{2016}]%
        {ohsaka2016dynamic}
\bibfield{author}{\bibinfo{person}{Naoto Ohsaka}, \bibinfo{person}{Takuya
  Akiba}, \bibinfo{person}{Yuichi Yoshida}, {and} \bibinfo{person}{Ken{-}ichi
  Kawarabayashi}.} \bibinfo{year}{2016}\natexlab{}.
\newblock \showarticletitle{Dynamic Influence Analysis in Evolving Networks}.
\newblock \bibinfo{journal}{\emph{Proc. VLDB Endow.}} \bibinfo{volume}{9},
  \bibinfo{number}{12} (\bibinfo{year}{2016}), \bibinfo{pages}{1077--1088}.
\newblock


\bibitem[\protect\citeauthoryear{Pearl}{Pearl}{1984}]%
        {pearl1984heuristics}
\bibfield{author}{\bibinfo{person}{Judea Pearl}.}
  \bibinfo{year}{1984}\natexlab{}.
\newblock \showarticletitle{Heuristics: Intelligent Search Strategies for
  Computer Problem Solving}.
\newblock  (\bibinfo{year}{1984}).
\newblock


\bibitem[\protect\citeauthoryear{Savitch}{Savitch}{1970}]%
        {savitch1970relationships}
\bibfield{author}{\bibinfo{person}{Walter~J. Savitch}.}
  \bibinfo{year}{1970}\natexlab{}.
\newblock \showarticletitle{Relationships between nondeterministic and
  deterministic tape complexities}.
\newblock \bibinfo{journal}{\emph{J. Comput. Syst. Sci.}} \bibinfo{volume}{4},
  \bibinfo{number}{2} (\bibinfo{year}{1970}), \bibinfo{pages}{177--192}.
\newblock


\bibitem[\protect\citeauthoryear{Schrijver}{Schrijver}{2003}]%
        {schrijver2003combinatorial}
\bibfield{author}{\bibinfo{person}{Alexander Schrijver}.}
  \bibinfo{year}{2003}\natexlab{}.
\newblock \bibinfo{booktitle}{\emph{Combinatorial Optimization: Polyhedra and
  Efficiency}}.
\newblock \bibinfo{publisher}{Springer Science \& Business Media}.
\newblock


\bibitem[\protect\citeauthoryear{Sharma, Kapoor, and Deshpande}{Sharma
  et~al\mbox{.}}{2015}]%
        {sharma2015greedy}
\bibfield{author}{\bibinfo{person}{Dravyansh Sharma}, \bibinfo{person}{Ashish
  Kapoor}, {and} \bibinfo{person}{Amit Deshpande}.}
  \bibinfo{year}{2015}\natexlab{}.
\newblock \showarticletitle{On greedy maximization of entropy}. In
  \bibinfo{booktitle}{\emph{ICML}}. \bibinfo{pages}{1330--1338}.
\newblock


\bibitem[\protect\citeauthoryear{Tang, Xiao, and Shi}{Tang
  et~al\mbox{.}}{2014}]%
        {tang2014influence}
\bibfield{author}{\bibinfo{person}{Youze Tang}, \bibinfo{person}{Xiaokui Xiao},
  {and} \bibinfo{person}{Yanchen Shi}.} \bibinfo{year}{2014}\natexlab{}.
\newblock \showarticletitle{Influence Maximization: Near-Optimal Time
  Complexity Meets Practical Efficiency}. In
  \bibinfo{booktitle}{\emph{SIGMOD}}. \bibinfo{pages}{75--86}.
\newblock


\bibitem[\protect\citeauthoryear{van~den Heuvel}{van~den Heuvel}{2013}]%
        {heuvel13complexity}
\bibfield{author}{\bibinfo{person}{Jan van~den Heuvel}.}
  \bibinfo{year}{2013}\natexlab{}.
\newblock \showarticletitle{The complexity of change}.
\newblock In \bibinfo{booktitle}{\emph{Surveys in Combinatorics 2013}}.
  Vol.~\bibinfo{volume}{409}. \bibinfo{pages}{127--160}.
\newblock


\bibitem[\protect\citeauthoryear{Vargas and Castells}{Vargas and
  Castells}{2011}]%
        {vargas2011rank}
\bibfield{author}{\bibinfo{person}{Sa{\'u}l Vargas} {and}
  \bibinfo{person}{Pablo Castells}.} \bibinfo{year}{2011}\natexlab{}.
\newblock \showarticletitle{Rank and relevance in novelty and diversity metrics
  for recommender systems}. In \bibinfo{booktitle}{\emph{RecSys}}.
  \bibinfo{pages}{109--116}.
\newblock


\bibitem[\protect\citeauthoryear{Vondr{\'a}k}{Vondr{\'a}k}{2010}]%
        {vondrak2010submodularity}
\bibfield{author}{\bibinfo{person}{Jan Vondr{\'a}k}.}
  \bibinfo{year}{2010}\natexlab{}.
\newblock \showarticletitle{Submodularity and Curvature: The Optimal Algorithm
  (Combinatorial Optimization and Discrete Algorithms)}.
\newblock \bibinfo{journal}{\emph{RIMS K{\^o}ky{\^u}roku Bessatsu}}
  \bibinfo{number}{23} (\bibinfo{year}{2010}), \bibinfo{pages}{253--266}.
\newblock


\bibitem[\protect\citeauthoryear{Wilhelm, Ramanathan, Bonomo, Jain, Chi, and
  Gillenwater}{Wilhelm et~al\mbox{.}}{2018}]%
        {wilhelm2018practical}
\bibfield{author}{\bibinfo{person}{Mark Wilhelm}, \bibinfo{person}{Ajith
  Ramanathan}, \bibinfo{person}{Alexander Bonomo}, \bibinfo{person}{Sagar
  Jain}, \bibinfo{person}{Ed~H. Chi}, {and} \bibinfo{person}{Jennifer
  Gillenwater}.} \bibinfo{year}{2018}\natexlab{}.
\newblock \showarticletitle{Practical diversified recommendations on {YouTube}
  with determinantal point processes}. In \bibinfo{booktitle}{\emph{CIKM}}.
  \bibinfo{pages}{2165--2173}.
\newblock


\bibitem[\protect\citeauthoryear{Yao, Fan, Zhao, Wan, Chang, and Xiao}{Yao
  et~al\mbox{.}}{2016}]%
        {yao2016tweet}
\bibfield{author}{\bibinfo{person}{Jin{-}ge Yao}, \bibinfo{person}{Feifan Fan},
  \bibinfo{person}{Wayne~Xin Zhao}, \bibinfo{person}{Xiaojun Wan},
  \bibinfo{person}{Edward~Y. Chang}, {and} \bibinfo{person}{Jianguo Xiao}.}
  \bibinfo{year}{2016}\natexlab{}.
\newblock \showarticletitle{Tweet Timeline Generation with Determinantal Point
  Processes}. In \bibinfo{booktitle}{\emph{AAAI}}. \bibinfo{pages}{3080--3086}.
\newblock


\end{thebibliography}
